\documentclass[3p,nopreprintline,sort&compress]{elsarticle}
\usepackage[colorlinks]{hyperref}
\usepackage{booktabs}
\usepackage{float}
\usepackage{caption}
\usepackage{amsmath, amsfonts, amsthm}
\usepackage[capitalise,noabbrev]{cleveref}
\usepackage{enumitem}
\usepackage{siunitx}
\usepackage{csvsimple}
\usepackage{booktabs}
\usepackage{bm}
\usepackage{tikz}
\usepackage{pgfplots}
\usepackage{pgfplotstable}
\usepackage{longtable}
\usepackage[ruled,vlined,linesnumbered]{algorithm2e}
\usepackage{algorithmic}
\usepackage{bbm}
\usepackage{pifont}

\newtheorem{theo}{Theorem}

\theoremstyle{definition}
\newtheorem{example}{Example}

\usetikzlibrary{pgfplots.groupplots}
\pgfplotsset{compat=1.6}

\definecolor{ClrFill1}{HTML}{e6ebf9}
\definecolor{ClrFill2}{HTML}{b897c7}
\definecolor{ClrFill3}{HTML}{824c9a}
\definecolor{ClrFill4}{HTML}{4e79c3}
\definecolor{ClrFill5}{HTML}{57a3ab}
\definecolor{ClrFill6}{HTML}{7eb975}
\definecolor{ClrFill7}{HTML}{d0b541}
\definecolor{ClrFill8}{HTML}{e67f32}
\definecolor{ClrFill9}{HTML}{ce2123}
\definecolor{ClrFill10}{HTML}{531914}

\definecolor{Clr1}{HTML}{1f77b4}
\definecolor{Clr2}{HTML}{ff7f0e}
\definecolor{Clr3}{HTML}{2ca02c}
\definecolor{Clr4}{HTML}{d62728}
\definecolor{Clr5}{HTML}{9467bd}
\definecolor{Clr6}{HTML}{8c564b}
\definecolor{Clr7}{HTML}{e377c2}
\definecolor{Clr8}{HTML}{7f7f7f}
\definecolor{Clr9}{HTML}{bcbd22}
\definecolor{Clr10}{HTML}{17becf}

\definecolor{TaskClrFill0}{HTML}{696969}
\definecolor{TaskClrFill1}{HTML}{1f77b4}
\definecolor{TaskClrFill2}{HTML}{a0522d}
\definecolor{TaskClrFill3}{HTML}{7f0000}
\definecolor{TaskClrFill4}{HTML}{aec7e8}
\definecolor{TaskClrFill5}{HTML}{808000}
\definecolor{TaskClrFill6}{HTML}{ff7f0e}
\definecolor{TaskClrFill7}{HTML}{3cb371}
\definecolor{TaskClrFill8}{HTML}{008080}
\definecolor{TaskClrFill9}{HTML}{4682b4}
\definecolor{TaskClrFill10}{HTML}{ffbb78}
\definecolor{TaskClrFill11}{HTML}{2ca02c}
\definecolor{TaskClrFill12}{HTML}{98df8a}
\definecolor{TaskClrFill13}{HTML}{daa520}
\definecolor{TaskClrFill14}{HTML}{7f007f}
\definecolor{TaskClrFill15}{HTML}{b03060}
\definecolor{TaskClrFill16}{HTML}{d62728}
\definecolor{TaskClrFill17}{HTML}{ff0000}
\definecolor{TaskClrFill18}{HTML}{00ced1}
\definecolor{TaskClrFill19}{HTML}{ff9896}
\definecolor{TaskClrFill20}{HTML}{6a5acd}
\definecolor{TaskClrFill21}{HTML}{9467bd}
\definecolor{TaskClrFill22}{HTML}{c5b0d5}
\definecolor{TaskClrFill23}{HTML}{8c564b}
\definecolor{TaskClrFill24}{HTML}{ba55d3}
\definecolor{TaskClrFill25}{HTML}{00fa9a}
\definecolor{TaskClrFill26}{HTML}{c49c94}
\definecolor{TaskClrFill27}{HTML}{e377c2}
\definecolor{TaskClrFill28}{HTML}{f08080}
\definecolor{TaskClrFill29}{HTML}{f7b6d2}
\definecolor{TaskClrFill30}{HTML}{ff00ff}
\definecolor{TaskClrFill31}{HTML}{7f7f7f}
\definecolor{TaskClrFill32}{HTML}{f0e68c}
\definecolor{TaskClrFill33}{HTML}{ffff54}
\definecolor{TaskClrFill34}{HTML}{c7c7c7}
\definecolor{TaskClrFill35}{HTML}{add8e6}
\definecolor{TaskClrFill36}{HTML}{ff1493}
\definecolor{TaskClrFill37}{HTML}{7fffd4}

\newcommand{\AbbrMPSOC}{MPSoC}

\newcommand{\AbbrWSS}{WSS}
\newcommand{\AbbrIPS}{IPS}
\newcommand{\ModelSumMax}{\texttt{SM}}
\newcommand{\ModelLR}{\texttt{LR}}
\newcommand{\PIIntervals}{processing-idling}

\newcommand{\AlgGA}{\texttt{GA}}
\newcommand{\ModelILPSMOrig}{\texttt{ILP-SM}}

\newcommand{\ModelQPLR}{\texttt{QP-LR-UB}}
\newcommand{\ModelILPFeasibility}{\texttt{ILP-FEAS}}

\newcommand{\ModelILPIdleMin}{\texttt{ILP-IDLE-MIN}}
\newcommand{\ModelILPIdleMax}{\texttt{ILP-IDLE-MAX}}

\newcommand{\ModelBlackBox}{\texttt{BB}}
\newcommand{\ModelBlackBoxSM}{\texttt{BB-SM}}

\newcommand{\ModelHeur}{\texttt{HEUR}}

\newcommand{\SetIntPos}{\mathbb{Z}_{>0}}
\newcommand{\SetIntNonNeg}{\mathbb{Z}_{\geq 0}}
\newcommand{\SetRealNonNeg}{\mathbb{R}_{\geq 0}}

\newcommand{\MF}{h}

\newcommand{\PTotal}{P}
\newcommand{\PIdle}{P_{\text{idle}}}

\newcommand{\ResSymb}{C}
\newcommand{\ResSet}{\mathcal{\ResSymb}}
\newcommand{\ResNum}{m}
\newcommand{\ResIdx}{k}
\newcommand{\Res}[1]{\ResSymb_{#1}}
\newcommand{\ResCap}[1]{c_{#1}}

\newcommand{\CoreIdx}{r}

\newcommand{\TaskSymb}{\tau}
\newcommand{\TaskSet}{\mathcal{T}}
\newcommand{\TaskNum}{n}
\newcommand{\TaskIdx}{i}
\newcommand{\Task}[1]{\TaskSymb_{#1}}
\newcommand{\TaskProc}[2]{e_{#1,#2}}

\newcommand{\TaskChar}[1]{\mathcal{F}^{(t)}_{#1}}

\newcommand{\TaskCoefSlope}[2]{a_{#1,#2}}

\newcommand{\TaskCoefOffset}[2]{o_{#1,#2}}

\newcommand{\WinSymb}{W}
\newcommand{\WinSet}{\mathcal{\WinSymb}}

\newcommand{\Win}[1]{\WinSymb_{#1}}
\newcommand{\WinNum}{q} 
\newcommand{\WinIdx}{j}
\newcommand{\WinLen}[1]{l_{#1}}

\newcommand{\VarAss}[3]{x_{#1,#2,#3}}

\begin{document}

\begin{frontmatter}

\title{Thermal Modeling and Optimal Allocation of Avionics Safety-critical Tasks on Heterogeneous MPSoCs}



\author[cvut]{Zdeněk Hanzálek}
\author[cvut]{Ondřej Benedikt}
\author[cvut]{Přemysl Šůcha}
\author[honeywell]{Pavel Zaykov}
\author[cvut]{Michal Sojka\corref{mycorrespondingauthor}} \ead{michal.sojka@cvut.cz}

\cortext[mycorrespondingauthor]{Corresponding author}
\address[cvut]{Czech Technical University in Prague, Jugoslávských partyzánů 1580/3, 160 00 Praha 6, Czech Republic}
\address[honeywell]{Honeywell International s.r.o., Brno, Czech Republic}

\begin{abstract}
  Multi-Processor Systems-on-Chip (MPSoC) can deliver high performance needed in many industrial domains, including aerospace. However, their high power consumption, combined with avionics safety standards, brings new thermal management challenges. This paper investigates techniques for offline thermal-aware allocation of periodic tasks on heterogeneous MPSoCs running at a fixed clock frequency, as required in avionics. The goal is to find the assignment of tasks to (i) cores and (ii) temporal isolation windows, as required in ARINC~653 standard, while minimizing the MPSoC temperature. To achieve that, we formulate a new optimization problem, we derive its
NP-hardness, and we identify its subproblem solvable in polynomial time.  Furthermore, we propose and analyze three power models, and integrate them within several novel optimization approaches based on heuristics, a black-box optimizer, and Integer Linear Programming (ILP). We perform the experimental evaluation on three popular MPSoC platforms (NXP i.MX8QM MEK, NXP i.MX8QM Ixora, NVIDIA TX2) and observe a difference of up to 5.5°C among the tested methods (corresponding to a 22\% reduction w.r.t. the ambient temperature). We also show that our method, integrating the empirical power model with the ILP, outperforms the other methods on all tested platforms.
\end{abstract}


\end{frontmatter}


\section{Introduction}

Nowadays, modern avionics systems are based on Integrated Modular Avionics (IMA), which simplifies software development and certification efforts~\cite{watkins07:trans_from_feder_avion_archit}. The ARINC~653 standard~\cite{ARINC653P1-4}, which specifies the Real-Time Operating System (RTOS) interface~\cite{2008:Prisaznuk}, has been widely accepted. One of the prominent concepts adopted by ARINC~653 is time-partitioned scheduling, which ensures the required separation of individual application components.

Implementing partitioned scheduling on new platforms naturally brings some challenges. To a certain extent, many of them have already been addressed, such as resource sharing~\cite{2019:Dugo}, scheduling~\cite{2021:Han}, and security~\cite{2022:Potteigerb}.  However, the increasing demand for computing power in avionics applications brings new challenges related to thermal properties. Usually, recent avionics systems utilize modern and powerful Multi-Processor System-on-Chips (\AbbrMPSOC)~\cite{2012:Salloum,2018:Lentaris} which require  careful thermal management.

It is well known that overheating negatively affects system reliability and safety \cite{2014:Lakshminarayanan}. Moreover, the rise in the operational temperature of on-chip components may lead to irreversible permanent failures \cite{2019:Paul}. On the other hand, reducing the on-chip temperature leads to a decrease in leakage power, which nowadays accounts for a significant part of modern \AbbrMPSOC{} power consumption \cite{2013:Zhuravlev,2017:Pi}.
Typically, the temperature of aircraft computing systems is kept within the allowed limits by means of an external cooling system, but its high cost, weight and/or low reliability may represent a weakness in the overall system design. Therefore, the motivation of this paper is to investigate whether a software-based technique could be used to replace or reduce the cooling system.
Many thermal-related issues have already been addressed individually, including power and thermal modeling~\cite{2006:Huang, 2017:Yoon}, design of energy-efficient real-time scheduling algorithms~\cite{2016:Zhou, 2016:Chien, 2018:Zhou}, energy-efficient execution of redundant, safety-critical tasks~\cite{2023:Wu,2024:Xu}, scheduling of real-time tasks with energy harvesting~\cite{2023:Schieber}, and study of thermal behavior of hardware platforms~\cite{2020:Li, 2022:Sadiqbatcha}.
However, the combination of ARINC-653 partitioned scheduling and thermal management has not yet been addressed.

Throughout this paper, we address the problem of thermally efficient task allocation under ARINC-653 temporal isolation windows (further called ``windows'' only) on heterogeneous \AbbrMPSOC{}s. We assume that the allocation is computed offline (in the design phase) and that Dynamic Voltage and Frequency Scaling (DVFS) is not used, as it increases failure rates~\cite{2023:Wu} and therefore is typically forbidden by typical safety certification requirements~\cite{2017:Fakih}. Since the combination of thermal-aware offline scheduling with time-partitioning constraints has not been addressed before, we propose several optimization methods to allocate periodic workload on heterogeneous \AbbrMPSOC{} while minimizing the steady-state on-chip temperature.
We focus on the relationship between the selected thermal model and the optimization method that integrates it, and how the inaccuracy of the former affects the efficiency of the latter.
As required by our industry partner, we emphasize a data-driven evaluation based on the measured characteristics of real physical platforms.
Therefore, we conduct a series of experiments using three hardware platforms (I.MX8QM MEK~\cite{imx8-mek}, I.MX8QM Ixora~\cite{2021:NXP-IMX8QM}, NVIDIA TX2~\cite{TX2devkit}) to assess the quality of the proposed methods.

In our experiments, we use an open-source ARINC-653-like Linux scheduler called DEmOS that we developed and publicly released (see \url{https://github.com/ CTU-IIG /demos-sched}). It provides independence from proprietary avionics RTOSes. Furthermore, we make all measured data, as well as optimization methods, publicly available at \url{https://github.com/benedond/safety-critical-scheduling}. This allows easy reproducibility of our results and enables the use of the proposed methods in an industrial context, possibly with different hardware platforms.

\noindent\textbf{Contributions.}
\begin{itemize}
    \item  We formally define an \textit{ARINC problem} which considers a simple \emph{thermal model} first. We derive its NP-hardness by a polynomial reduction from a PARTITION problem. Furthermore, we derive the polynomial time complexity of the \textit{ARINC problem} with fixed lengths of windows by reduction to a Minimum Cost Flow problem,
    
    \item  We propose multiple optimization methods to tackle the problem of thermal-aware task allocation on heterogeneous \AbbrMPSOC{} under ARINC-653 temporal isolation constraints, including two informed methods based on mathematical programming, two informed methods based on genetic algorithms, one local informed heuristic, and two uniformed heuristics.

    \item  We analyze the trade-offs between the accuracy of the power model and the performance of the optimization method by integrating several power models with the proposed optimization methods.

    \item We conduct physical experiments on three different hardware platforms demonstrating practical applicability of the results and showing that, across all tested scenarios and platforms, the empirical Sum-Max power model (\ModelSumMax{}) integrated within an Integer Linear Programming formalism outperforms the other methods in terms of the thermal objective.

    \item We make all measured data and source code of proposed optimization methods publicly available.

\end{itemize}

This paper substantially extends our preliminary study \cite{2021:Benedikt} by: (i) introducing more optimization methods; (ii) extending the evaluation of the empirical \ModelSumMax{} power model; (iii) introducing a linear regression-based power model for comparison; (iv) conducting experiments on three hardware platforms; (v) introducing new benchmarks based on the industrial standard EEMBC~Autobench~2.0; and (vi) overall extending the scope of the experiments.

\textbf{Outline.} The rest of this paper is organized as follows. \Cref{sec:proposed-methodology} summarizes the related work. \Cref{sec:system-model-and-pd} describes the system and formalizes the scheduling problem definition. Thermal and power modeling with regard to the allocation of the safety-critical tasks and implementation of specific models is addressed in \Cref{sec:thermal-modeling}. The main outcome -- the integration of the thermal models with the optimization procedures is discussed in \Cref{sec:solution-methods}.
Experimental evaluation is split between \cref{sec:hw-and-benchmarks}, where we describe used hardware platforms and benchmarks, and \Cref{sec:experiments}, where we compare the proposed methods. Finally, \Cref{sec:conclusion} concludes the paper.

\section{Related Work}
\label{sec:proposed-methodology}

To the best of our knowledge, our paper is the first to address the unique challenge of combining avionics-inspired time-partitioned scheduling of safety-critical workloads with thermal issues on real heterogeneous hardware platforms. We have not come across any other work that has tackled this specific problem.
In~\Cref{sec:rw-temporal-isolation}, we review the previous research on the use of windows as a method for reducing interference in real-time safety-critical systems, and in~\Cref{sec:rw-thermal-aware-schedulibng} the research on thermal-aware scheduling and optimization.

\subsection{Temporal Isolation}
\label{sec:rw-temporal-isolation}

 Temporal isolation, which is a key concept in ARINC-653 standard, ensures that different tasks within the avionics system are executed in a predictable and deterministic manner, even in the presence of hardware or software failures. Additionally, temporal isolation ensures that the system can meet real-time performance requirements, such as worst-case response times, even in the presence of changes in the system's workload or environment.

Several works address the problem of scheduling under temporal isolation constraints. 
Han et al. proposed a model-based optimization method for addressing the temporal isolation of systems using the ARINC 653 standard~\cite{2021:Han}. Their method based on heuristic search was intended to minimize the processor occupancy of the system, making it possible to accommodate additional application workload.
Tama\c{s}-Selicean and Pop researched the time-partitioned scheduling of safety-critical and non-critical applications in their paper~\cite{2011:Selicean}. They proposed an optimization approach based on Simulated Annealing to determine the sequence and length of partitions, assuming a fixed mapping of tasks to processing elements. In their subsequent work~\cite{2015:Selicean}, they extended the problem to also optimize the allocation of tasks to processing elements, considering different criticality levels of the applications. To solve this problem, they proposed an optimization approach based on Tabu Search. Waszniowski et al.~\cite{2009:Waszniowski} used timed automata to verify a distributed fault-tolerant real-time application.
Chen et al.~\cite{2016:Chen} investigated the scheduling of independent partitions in the context of IMA systems. They proposed a Mixed Integer Linear Programming (MILP) model to represent the schedulability constraints of the independent partitions, where each partition was modeled as a non-preemptive periodic task. Additionally, the authors proposed a heuristic approach to determine a start time and processor allocation for each partition.

Even though the aforementioned works address scheduling under temporal isolation constraints, none of them aims for energy-efficient schedule. As energy consumption (thermal-efficiency) and timeliness are conflicting objectives, direct application of the proposed methods in the context of thermal optimization is not easily possible. We aim at designing such methods that would respect the temporal constraints while minimizing the on-chip temperature.
 
\subsection{Thermal-aware Scheduling}
\label{sec:rw-thermal-aware-schedulibng}

Thermal-aware and energy-efficient scheduling for real-time systems has been studied for many years~\cite{chen07:energ_effic_sched_for_real,2011:Fisher,2016:Gerards}. Even though the individual approaches differ in many aspects, there are several common steps -- namely \emph{benchmarking}, \emph{optimization}, and  \emph{evaluation}.
The decisions made at each step relate to the other steps and have an influence on the overall properties of achieved results. In the following, we describe each step and related decisions in more detail and review relevant literature.

\begin{figure}[htb]
    \centering
    \resizebox{\textwidth}{!}{
        \newcommand{\Border}[3]{
    \draw[thick] ($(#1.north west)+(-0.3,0.3)$)  rectangle ($(#2.south east)+(0.3,-0.3)$);
	\node at ($(#1.north) + (0,0.3) $) [anchor=south] {\footnotesize  \textbf{#3}};
}

\begin{tikzpicture}

\tikzset{RNode/.style={rectangle, draw, font=\footnotesize, minimum size=1cm, text centered, inner sep=2pt, text width=2cm}}
\tikzset{Arrow/.style={-latex, thick}}
\tikzset{Desc/.style={font=\scriptsize, pos=0.5}}

\node[RNode] (NTaskSet) at (0,0) {Task set};
\node[RNode] (NHW) [below of=NTaskSet, yshift=-0.4cm] {Hardware platform};
\Border{NTaskSet}{NHW}{Benchmarking};

\node[RNode] (NAlg) [right of=NTaskSet, xshift=6cm] {Algorithm};
\node[RNode] (NModel) [below of=NAlg, yshift=-0.4cm] {Thermal model};
\Border{NAlg}{NModel}{Optimization};
\draw[Arrow] ($(NAlg.south) - (0.2,0)$) -- ($(NModel.north) - (0.2,0)$);
\draw[Arrow] ($(NModel.north) + (0.2,0)$) -- ($(NAlg.south) + (0.2,0)$);

\draw[Arrow] (NTaskSet.east) -- ($(NAlg.west)-(0.3,0)$) node[Desc, anchor=south] {Task Characteristics};
\draw[Arrow] (NHW.east) -- ($(NModel.west)-(0.3,0)$) node[Desc, anchor=north] {Platform Characteristics};

\node[RNode,text width=3cm] (NEval) [right of=NAlg, xshift=4.5cm] {Evaluation on the physical platform};
\node[RNode,text width=3cm, minimum size=0.5cm] (NData) [below of=NEval, yshift=0.00cm] {Data Collection};
\node[RNode,text width=3cm, minimum size=0.5cm] (NStats) [below of=NData, yshift=0.25cm] {Statistical evaluation};
\Border{NEval}{NStats}{Evaluation};
\draw[Arrow] ($(NAlg.east) + (0.3,-0.65)$) -- ($(NEval.west) - (0.3,0.65) $) node[Desc,anchor=south] {Schedule};

\node[RNode] (NInst) [right of=NTaskSet, xshift=2.3cm, yshift=-0.7cm] {Problem instance};
\draw[Arrow] (NInst.east) -- ($(NAlg.west) - (0.3, 0.7)$);

\draw[thick,dashed] (1.7,-3) -- (1.7,2);
\draw[Arrow] (1.7,-2.5) -- (0,-2.5) node[Desc, anchor=north, pos=0.7] {Done only once};
\draw[Arrow] (1.7,-2.5) -- (4,-2.5) node[Desc, anchor=north, pos=0.7] {Done for each instance};

\end{tikzpicture}
    }
    \caption{Three steps (bechmarking, optimization and evaluation) towards thermally efficient scheduling.}
    \label{fig:optimization-procedure}
\end{figure}
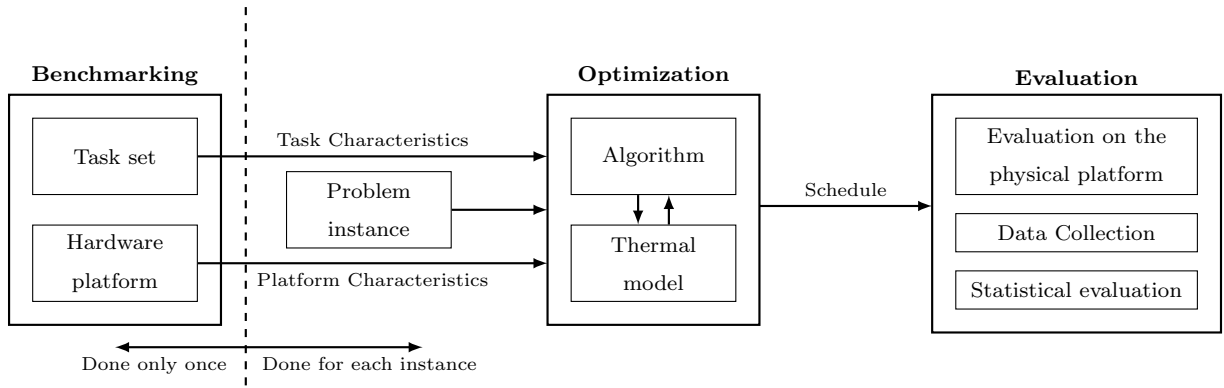

\paragraph{Benchmarking} Any thermal-aware optimization algorithm requires information about the platform and tasks to be executed, referred to as \emph{platform characteristics} and \emph{task characteristics}. 
 Platform characteristics describe the thermal and power behavior without relation to the workload and might include the area of the chip, power consumption w.r.t. selected frequency, and thermal conductances and capacitances. These parameters are often taken from technical documentation \cite{2019:Perez, nvidiaJetsonTX2Series2019} or pre-defined configurations provided with the simulation software like HotSpot \cite{2018:Manna, 2018:Kanduri}. However, these may not accurately reflect reality as thermal parameters are often influenced by the Printed Circuit Board (PCB) layout chosen by a particular board manufacturer and manufacturing variations. To obtain more accurate parameters, benchmarking can be used \cite{2020:Sojka,2022:Ara}.
Task characteristics allow for distinctions between workloads and their thermal effects, and are typically obtained through benchmarking.
However, the complexity of task characteristics can vary; some works assume all tasks to be identical \cite{2014:Chen}, while others assume a single numerical coefficient~\cite{2020:Zhou, 2018:Manna}, multiple coefficients \cite{2021:Benedikt,2016:Salami,2019:Balsinia}, or even very complex characteristics obtained, e.g., from CPU performance counters~\cite{2021:Sahid}, or neural networks~\cite{2018:Zhang}.
Care must be taken when selecting characteristics as their benchmarking can be time-consuming. In this paper, we show that even simple characteristics can be sufficient for significant temperature reduction.

\paragraph{Optimization}
The optimization procedure integrates the scheduling algorithm and the thermal model to allocate and schedule tasks while ensuring all other timing, thermal, and resource constraints are met. For safety-critical applications, such as in avionics, offline algorithms are often required~\cite{2002:Cofera}.
The thermal model predicts the evolution of on-chip temperature in time, based on the system state and workload. From the viewpoint of the systems dynamics, the thermal model can be transient-state~\cite{2022:Ara,2011:Fisher,2016:Zhou,2018:Li} or steady-state~\cite{2015:Zhu,2016:Salami,2018:Manna,2021:Abdollahi}, where the former is more general, while the latter is simpler to implement. According to Chantem et al.~\cite{2011:Chantem}, the steady-state model is sufficient if the temporal parameters of the workload are short enough, which is the case for our target applications in avionics. 
From the spacial point of view, the thermal model can be single-output~\cite{2021:Benedikt,2018:Jiang}, which provides a prediction for a single thermal node only, or multi-output~\cite{2022:Ara,2021:Abdollahi,2019:Li}.
As we experimentally show, it is quite difficult to distinguish between the temperatures of the individual computing elements of our tested platforms. Therefore, we use a single-output thermal model.
Considering the optimization procedures for task scheduling and allocation, many different approaches have been studied. Due to the inherent complexity of thermal-aware scheduling, authors often rely on local or greedy heuristics~\cite{2021:Abdollahi, 2019:Paul, 2019:Li, 2015:Zhu,2015:Kuo,2016:Zhou,2024:Xu}. Other approaches include meta-heuristics~\cite{2018:Manna, 2019:Cao,2018:Liu}, or formulations based on mixed-integer linear programming~\cite{2018:Manna, 2021:Abdollahi,2011:Chantem,2018:Jiang}. In this paper, we use all these three approaches and compare them.
Another aspect to consider is the interaction between the scheduling algorithm and the thermal model. In the literature, we can find the following approaches: (i) The thermal model is used only to validate whether the schedule provided by the scheduling algorithm can be executed under the given thermal constraints or not \cite{2022:Ara}; (ii)
In more complex cases, the loop is closed, and the information about  thermal constraints violation is fed back to the scheduling algorithm, which, in turn, tries to rebuild the schedule~\cite{2021:Abdollahi,2019:Perez}; (iii)
Finally, the thermal model and the thermal constraints can be integrated directly within the scheduling algorithm, thus providing the most integrated solution \cite{2011:Chantem,2021:Benedikt}. In this paper, we implement and evaluate all three approaches.

\paragraph{Evaluation} Evaluation of the properties and performance of the resulting schedule is typically conducted  (i) in a simulator or (ii) by evaluation on a physical platform. Simulation-based approaches are found more often~\cite{2021:Abdollahi, 2019:Paul, 2019:Li, 2018:Zhou, 2018:Kanduri,2011:Chantem,2015:Kuo,2018:Jiang,2023:Wu,2024:Xu}. The reasons justifying this approach include simpler execution of the experiments and better reproducibility of the results. On the other hand, simulation is always based on models, which might fail to properly capture all details of the hardware platform. Thus some authors evaluate thermal effects of the schedules experimentally on real hardware~\cite{lucasMEMPowerDataAwareGPU2019, mantovaniPerformancePowerAnalysis2018, leeThermalAwareSchedulingIntegratedCPUs2019,2015:Pallister,2022:Ara}.
Among those papers, the majority uses just one hardware platform~\cite{leeThermalAwareSchedulingIntegratedCPUs2019,lucasMEMPowerDataAwareGPU2019,2022:Ara}, and those who evaluate on more platforms do not combine that with optimization techniques.
We believe that experimenting on real hardware provides a more accurate representation of the thermal behavior of a system due to its ability to reflect real-world conditions, take into account hardware variations, and allows observing interactions with other components. Therefore, we follow the experimental path in this paper.
Paper in hand is based on tight integration of the hardware and optimization algorithms, evaluates on three platforms, and addresses all aspects of the pipeline illustrated in \Cref{fig:optimization-procedure}.

\section{Goal and Problem Formalization} \label{sec:system-model-and-pd}

In this section, we define the goal of the thermal optimization and formalize the scheduling problem of the thermal-aware safety-critical tasks (further called ``tasks'' only) allocation on \AbbrMPSOC{} under temporal isolation constraints. 

\subsection{Goal}

We want to find an assignment of tasks to CPUs of a heterogeneous multi-core platform together with an allocation of tasks to the windows such that the steady-state temperature of the platform is minimized.

There are at least three factors that make this problem complex.
(i) All the tasks must be scheduled within the pre-defined major frame, which repeats indefinitely; therefore, any task allocation with the makespan exceeding the major frame is not feasible.
(ii) The number of used windows and their lengths are not known a priori.
(iii) The steady-state temperature depends on the thermal interference of the tasks running in parallel on different CPUs.

\subsection{Problem Statement} 

We define our problem and its parameters as follows.

\paragraph{Processing Elements} We assume a heterogeneous architecture, i.e., \AbbrMPSOC{} having $\ResNum$ computing \emph{clusters} denoted by $\{ \Res{1}, \Res{2}, \dots, \Res{\ResNum} \} = \ResSet$. Different clusters have different hardware architectures and can have different frequencies, therefore, they are unrelated resources in the scheduling terminology.
Cluster $\Res{\ResIdx}$ has $\ResCap{\ResIdx} \in \SetIntPos$ cores, which are assumed to be identical. All cores in the cluster share the same clock frequency, which we assume to be fixed due to typical safety requirements~\cite{2017:Fakih}. 

\paragraph{Tasks} We assume a set of independent, non-preemptive, periodic tasks \hbox{$\TaskSet = \{\Task{1}, \Task{2}, \dots, \Task{\TaskNum} \}$}. Each task represents a single-threaded safety-critical process that needs to be executed on one of the platform cores. By $\TaskProc{\TaskIdx}{\ResIdx} \in \SetIntPos$ we denote the \emph{worst-case execution time} of task $\Task{\TaskIdx} \in \TaskSet$ on cluster $\Res{\ResIdx} \in \ResSet$. All tasks are ready at time $0$ and have a common period $\MF \in \SetIntPos$, which is called a \emph{major frame length}. We assume that the deadline of each task is equal to $\MF$.  In this section, we consider a simple \emph{thermal model} related to the energy consumption as follows:  by $d_{{\TaskIdx},{\ResIdx}} \in \SetIntPos$ we denote the  \emph{energy consumption} of task $\Task{\TaskIdx}  \in \TaskSet$ when executed on cluster $\Res{\ResIdx} \in \ResSet$. 

\paragraph{Arinc Temporal Isolation} The temporal isolation of the tasks is ensured by so-called \emph{windows}, which are non-overlapping intervals partitioning the major frame (see \cref{fig:example-schedule}) as in the ARINC-653 standard; we discuss more details in \cite{2021:Benedikt}. We denote the set of such windows as $\WinSet = \{\Win{1}, \Win{2}, \dots, \Win{\WinNum}\}$. Length of window $\Win{\WinIdx} \in \WinSet$ is denoted by $\WinLen{\WinIdx} \in \SetIntNonNeg$. Each task needs to be assigned to a single window, within which it will be executed on one core. At most one task per core can be executed within each window. 
Note that the number of available windows $\WinNum$ is upper-bounded by $\TaskNum$ (a situation when only one task will be present in each window) and lower-bounded by $  \lceil \TaskNum / \sum_k \ResCap{\ResIdx} \rceil$ (a situation when every core is used in every window). 
The number of used windows (i.e., the ones with at least one task assigned) is not known a priori, and it is a part of the decision.

\paragraph{Criterion function} Let $\TaskSet_{\ResIdx}$ be the set of tasks assigned to cluster $\Res{\ResIdx}$. In this section, we consider a simple \emph{thermal model}: the objective is to minimize the energy consumption $E= \sum\limits_{\Res{\ResIdx} \in \ResSet} \sum\limits_{\Task{\TaskIdx} \in \TaskSet_{\ResIdx}} d_{{\TaskIdx},{\ResIdx}}$ of all tasks in the major frame. 

\vspace{0,5 cm}
Therefore we define an\textit{ ARINC problem}  as follows: 

\vspace{0,5 cm}
\noindent\fbox{
\scalebox{1.0}[1.0]{
\begin{minipage}[l]{0.95\linewidth}
Input: $n$ tasks, $m$ clusters with $\ResCap{\ResIdx}$ cores on cluster $\Res{\ResIdx}$, $q$ windows  $\Win{1}, \dots, \Win{\WinNum}$, major frame length~$\MF$, processing time $\TaskProc{\TaskIdx}{\ResIdx}$ and energy consumption $d_{{\TaskIdx},{\ResIdx}}$ of task $i$ on cluster $k$.

Constraints: 
\begin{itemize}
    \item all tasks are assigned,
    \item at most $\ResCap{\ResIdx}$ tasks are assigned to cluster $\Res{\ResIdx}$ in each window,
    \item each window is at least as long as the longest task assigned to it, i.e., for each task $\Task{\TaskIdx}$ assigned to cluster $\Res{\ResIdx}$ and window $\Win{\WinIdx}$ holds $\WinLen{\WinIdx} \geq \TaskProc{\TaskIdx}{\ResIdx}$, and
    \item the total length of all windows is at most equal to the major frame length $\MF$, i.e., $\sum\limits_{\Win{\WinIdx} \in \WinSet} \WinLen{\WinIdx} \leq \MF$.
\end{itemize} 

Objective: Find an assignment of tasks to clusters and windows such that all constraints are satisfied and $E= \sum\limits_{\Res{\ResIdx} \in \ResSet} \sum\limits_{\Task{\TaskIdx} \in \TaskSet_{\ResIdx}} d_{{\TaskIdx}{\ResIdx}}$ is minimum.
\end{minipage}
}
}
\vspace{0,5 cm}

While obtaining the assignment of tasks to clusters and windows, it is trivial to derive the schedule (i.e., assignment of tasks to cores in time) since every cluster-window pair contains at most $\ResCap{\ResIdx}$ tasks and, therefore, it is trivial to assign at most one task to each core in each window. An order of windows is arbitrary. 

Used notation is illustrated in \Cref{fig:example-schedule} on a schedule of seven tasks.

\begin{figure}[htb]
    \centering
    \begin{tikzpicture}[xscale=1.5, yscale=0.5]

\draw[thick,blue,dotted] (0,-4) -- (0,3.5);
\draw[thick,blue,dotted] (1.5,-4) -- (1.5,2.5);
\draw[thick,blue,dotted] (4,-4) -- (4,2.5);
\draw[thick,blue,dotted] (6,-4) -- (6,3);

\node[anchor=south] at (0,3.5) {0};
\node[anchor=south] at (6,3.5) {$\MF$};

\draw[thick, latex-latex] (0,3.5) -- node[anchor=south] {Major frame} (6,3.5);

\begin{scope}[font=\scriptsize, yshift=-0.5cm]
    \draw[thick, latex-latex] (0,3) -- node[anchor=south] {$\WinLen{1}$} (1.5,3);
    \draw[thick, latex-latex] (1.5,3) -- node[anchor=south] {$\WinLen{2}$} (4,3);
    \draw[thick, latex-latex] (4,3) -- node[anchor=south] {$\WinLen{3}$} (6,3);
\end{scope}

\node[anchor=north,font=\scriptsize] at (0.75,-4) {Window $\Win{1}$};
\node[anchor=north,font=\scriptsize] at (2.75,-4) {Window $\Win{2}$};
\node[anchor=north,font=\scriptsize] at (5,-4) {Window $\Win{3}$};

\node[font=\footnotesize,anchor=east] at (0,1.5) {Core 0};
\node[font=\footnotesize,anchor=east] at (0,0.5) {Core 1};
\node[font=\footnotesize,anchor=east] at (0,-0.5) {Core 2};
\node[font=\footnotesize,anchor=east] at (0,-1.5) {Core 3};
\node[font=\footnotesize,anchor=east] at (0,-2.5) {Core 4};
\node[font=\footnotesize,anchor=east] at (0,-3.5) {Core 5};

 \draw [thick,decorate,decoration={brace,amplitude=5pt,mirror},xshift=0pt,yshift=0pt] (-0.75, 2) -- (-0.75,-1.95) node [black,midway,xshift=-0.9cm,align=center, font=\footnotesize] { Cluster $\Res{1}$\\$\ResCap{1} = 4$};
 
 \draw [thick,decorate,decoration={brace,amplitude=5pt,mirror},xshift=0pt,yshift=0pt] (-0.75, -2.05) -- (-0.75,-4) node [black,midway,xshift=-0.9cm,align=center, font=\footnotesize] { Cluster $\Res{2}$\\$\ResCap{2} = 2$}; 

\draw[rounded corners,fill=blue!10] (0,1.1) rectangle node[font=\scriptsize, align=center] {$\Task{2}$} (1.5,1.9);
\draw[rounded corners,fill=blue!10] (0,0.1) rectangle node[font=\scriptsize, align=center] {$\Task{4}$} (1.3,0.9);
\draw[rounded corners,fill=blue!10] (0,-1.9) rectangle node[font=\scriptsize, align=center] {$\Task{7}$} (1.45,-1.1);
\draw[rounded corners,fill=blue!10] (0,-3.9) rectangle node[font=\scriptsize, align=center] {$\Task{3}$} (0.5,-3.1);

\draw[rounded corners,fill=blue!10] (1.5,0.1) rectangle node[font=\scriptsize, align=center] {$\Task{6}$} (4,0.9);
\draw[rounded corners,fill=blue!10] (1.5,-2.9) rectangle node[font=\scriptsize, align=center] {Task $\Task{1}$} (3,-2.1);
\draw[rounded corners,fill=blue!10] (1.5,-0.9) rectangle node[font=\scriptsize, align=center] {$\Task{5}$} (3.75,-0.1);

\draw[latex-latex] (1.5,-3.2) -- node[anchor=north] {\footnotesize $\TaskProc{1}{2}$}(3,-3.2);


\end{tikzpicture}%
    \caption{Schedule of seven tasks on six cores and three windows}
    \label{fig:example-schedule}
\end{figure}

\subsubsection{Problem Complexity} \label{sssec:complexity}

\begin{theo}
\label{theo:ARINC_NP_hard}
The \textit{ARINC problem} is at least weakly NP-hard. 
\end{theo}

\begin{proof}
In order to derive the complexity of the \textit{ARINC problem}, we formulate its decision version while extending the problem input by maximum energy $D$. The decision question is the following: Is there any feasible assignment of tasks to clusters and windows such that all constraints of the \textit{ARINC problem} are satisfied and $\sum\limits_{\Res{\ResIdx} \in \ResSet} \sum\limits_{\Task{\TaskIdx} \in \TaskSet_{\ResIdx}} d_{{\TaskIdx},{\ResIdx}} \leq D$? 

In the following, we show that the decision version of the \textit{ARINC problem} is NP-complete.

Consider a \textit{PARTITION problem with an equal number of items}: 

\vspace{0,5 cm}\noindent
\fbox{
\scalebox{1.0}[1.0]{
\begin{minipage}[l]{0.95\linewidth}

Input: Finite set $A$ of items with positive integer sizes $b_1,\ldots,b_{|A|}$ and positive integer $B=\sum_{i \in A}b_i/2$.

Question: Is there a subset $A' \subseteq A$ such that  $ \sum_{i \in A'} b_i = B$ and $|A'|=|A|/2$.
\end{minipage}
}
}
\vspace{0,5 cm}

This PARTITION problem is known to be weakly NP-complete \cite{1990:Garey} 
and it is often used to derive a complexity of scheduling problems  \cite{2024:Hanen}.
Now consider the following polynomial reduction of the \textit{PARTITION problem with an equal number of items} to the decision version of the \textit{ARINC problem}:

We consider two clusters $\Res{1}, \Res{2}$ with one core on each cluster, thus $\ResCap{1}=\ResCap{2}=1$. We consider  $n=|A|$ tasks where task $i$ is related to item $i$, such that, execution times are $\TaskProc{\TaskIdx}{1}=b_i$, $\TaskProc{\TaskIdx}{2}=1$ and energy consumptions are $d_{{\TaskIdx},{1}}=B-b_i$, $d_{{\TaskIdx},{2}}=B$. The major frame length $\MF=B$, maximum energy $D=(n-1)B$, and number of available intervals $q=|A|/2=n/2$.

We claim that since we have $n/2$ intervals and two cores (one in each cluster), then the algorithm solving the decision version of the \textit{ARINC problem} assigns exactly one task on each core in each window. Therefore, we have $\TaskSet_{1}$ set of $n/2$ tasks on the first core corresponding to the subset $A'$. The items in $A'$ have the size less than or equal to $B$ since $\sum\limits_{i \in A'} b_i =  \sum\limits_{\Task{\TaskIdx} \in \TaskSet_{1}} e_{{\TaskIdx},{1}} \leq \MF= B$. 

Consequently, we have $n/2$ tasks on the second core, and its energy consumption is equal to $nB/2$. Thus energy consumption of the first machine is less than $(n-1)B - nB/2 = nB/2 - B$. Formally:

\begin{equation}  \label{eq:proof_energy}
\begin{aligned}
    \sum\limits_{\Task{\TaskIdx} \in \TaskSet_{1}} d_{{\TaskIdx},{1}} = \sum\limits_{\Task{\TaskIdx} \in \TaskSet_{1}} B - b_i & \leq 
    nB/2 - B
    \\
    nB/2 - \sum\limits_{\Task{\TaskIdx} \in \TaskSet_{1}} b_i & \leq 
    nB/2 - B
    \\
   \sum\limits_{\Task{\TaskIdx} \in \TaskSet_{1}} b_i & \geq 
    B
\end{aligned}
\end{equation}

With the conclusion of the previous paragraph, we obtain $\sum\limits_{i \in A'} b_i = B$. Thus, any yes answer to the \textit{ARINC decision problem} implies a yes answer to the \textit{PARTITION problem with an equal number of items}.

Conversely, a subset $A'$ obtained from the \textit{PARTITION problem with an equal number of items} gives a feasible schedule on the first core since  $\sum\limits_{i \in A'} b_i = B = \MF$, and we can always process the remaining $n/2$ tasks corresponding to items not in $A'$ on the second core. Energy consumption on both cores is also satisfied since $ nB/2 - \sum\limits_{\Task{\TaskIdx} \in \TaskSet_{1}} b_i + nB/2 = nB-B = D$. 

Therefore, the decision version of the \textit{ARINC problem} is at least weakly NP-complete, and the \textit{ARINC problem} itself is at least weakly NP-hard.
\end{proof}

Let us further study, which part of the problem is polynomialy solvable. 
Therefore, we define a ``fixed version of the \textit{ARINC problem}'' as the  \textit{ARINC problem} where the length 
$l_j$ is given for each window $\Win{\WinIdx}$ in the problem input.

\begin{theo}
\label{theo:ARINC_fixed_polynomial}
The fixed version of the \textit{ARINC problem} is solvable in polynomial time. 
\end{theo}

\begin{proof}
Consider well-known \textit{Minimum Cost Flow problem}: 

\vspace{0,5 cm}\noindent
\fbox{%
\scalebox{1.0}[1.0]{
\begin{minipage}[l]{0.96\linewidth}

Input:  directed graph $G$, upper bound $u_g \in \SetRealNonNeg$ and cost $p_g \in \mathbb{R}$ of every edge $g$ of graph $G$, balance~$b_v \in \mathbb{R}$ of every vertex $v$ of graph $G$. Balances satisfy $\sum_{v} b_v = 0$.

Objective:  Find the feasible flow $f_g \in \SetRealNonNeg$ for every edge $g$ of graph $G$, that minimizes $\sum_{g} f_g \cdot p_g$ and satisfies all the constraints, or decide that it does not exist. 

Constraints: 
\begin{itemize}
    \item the upper bound is respected, i.e., $f_g\leq u_g$ for every edge $g$ of graph $G$,
    \item the first Kirchhof's law is respected, i.e.,  $\sum_{g \in \delta^+(v)} f_g - \sum_{g \in \delta^-(v)} f_g = b_v$ for every vertex $v$ of graph $G$, where $\delta^+(v)$ is a set of edges leaving vertex $v$ and $\delta^-(v)$ is a set of edges entering vertex~$v$.
\end{itemize}
\end{minipage}
}
}
\vspace{0,5 cm}

We will formulate the fixed version of the \textit{ARINC problem} as the {Minimum Cost Flow problem} in the following way:

1) Create $\TaskNum$ tasks vertices $\Task{i}$ with balances $b_{\Task{i}}=1$.


2) Create $\WinNum \cdot \ResNum$ vertices $WC_{\WinIdx\ResIdx}$ for each window $\Win{\WinIdx}$ and cluster 
$\Res{\ResIdx}$. Create an arc from vertex $\Task{i}$ to $WC_{\WinIdx\ResIdx}$ with upper bound $1$ exists if and only if $e_{i,\ResIdx} \leq l_\WinIdx$, i.e.,  the duration of task $\Task{i}$ on cluster $\Res{\ResIdx}$ fits in window $\Win{\WinIdx}$. A unit flow on this arc expresses that task $\Task{i}$ is assigned on cluster $\Res{\ResIdx}$ in window $\Win{\WinIdx}$. The balance of each $WC_{\WinIdx\ResIdx}$ vertex is equal to zero.

3) Create a sink vertex, and arcs from each $WC_{\WinIdx\ResIdx}$ to the sink vertex with upper bound $\ResCap{\ResIdx}$ indicating that the number of tasks assigned to each cluster does not exceed the number of cores. The balance of the sink vertex is equal to $- \TaskNum$.

If there is a solution to the \textit{Minimum Cost Flow problem}, we will obtain an integer solution due to the integer upper bounds. Therefore, every arc from vertex $\Task{i}$ to $WC_{\WinIdx\ResIdx}$ will have a flow equal to 0 or 1. The well-known cycle-canceling algorithm can provide an optimal solution in polynomial time. 
\end{proof}

\subsubsection{Complex Thermal Model}


In order to minimize the steady-state temperature of the real platform as formulated by the goal at the beginning of this section, the task \emph{thermal model} and \emph{criterion function} will get more complex in Section~\ref{sec:thermal-modeling}.  We assume that both the platform and tasks have an influence on the steady-state temperature. Therefore, Section \ref{sec:thermal-modeling} defines so-called \emph{platform characteristic} and \emph{task characteristic} that are obtained by benchmarking and characterize the platform's thermal/power behavior. We consider two different models (namely a \emph{Sum-Max Model} and \emph{Linear Regression Model}). Both models can be easily reduced to the simple thermal model considered in this section, and thus optimization problems based on those models are NP-hard as well.  

\section{Thermal and Power Modeling} \label{sec:thermal-modeling}

A lot of attention must be paid when designing a thermal model.
In our view, the main aspects of the thermal model to be considered and balanced are its \emph{simplicity} and \emph{accuracy}. A simpler model is easier to integrate with the optimization procedures. Also, it takes less effort to identify its parameters. However, a too simple model fails to predict the system's behavior accurately. Therefore, the trade-off between simplicity and accuracy needs to be taken into account.

In this paper, we decided to use a \emph{steady-state}, \emph{single-output} thermal model. This decision is based on the fact that typical lengths of the major fame used in avionics applications are less than one second, while the assumed hardware platforms have relatively slow thermal dynamics and sufficiently low differences in temperature at different places on the chip. This decision is experimentally justified in \cref{sec:thermal-model-analysis}.

In the following \Cref{sec:thermal-to-power}, we discuss the transition from thermal to power modeling, and then, in \Cref{sec:thermal-model-implementation}, we define the power models that we later integrate with the optimization procedures.

\subsection{Transition from Thermal to Power Model} \label{sec:thermal-to-power}

The single-output steady-state thermal model of an MPSoC adopted in this paper is tightly related to a simpler model based on average power consumption. In some sense, they can be used interchangeably, but from the practical viewpoint, the difference is very important. Deviations in the ambient temperature cause deviations in the steady-state temperature. On the other hand, measuring the power input is usually more stable and easily reproducible. Further, the time needed to reach stabilized temperature can be rather long, whereas the power measurements reflect the immediate state. 

Widely used methodology for creating thermal models of \AbbrMPSOC{} relies on Resistance-Capacitance (RC) thermal networks \cite{2006:Huang, 2014:Pagani, 2019:Perez}. The system is modeled as a set of thermal nodes, that are interconnected via thermal conductances and associated with thermal capacitances. The relation between the temperature of every thermal node, its power consumption, and the ambient temperature can then be expressed by a set of differential equations~\cite{2014:Pagani}:

\begin{equation} \label{eq:rc-system}
    \bm{A} \bm{T}^{\prime} + \bm{B}\bm{T} = \bm{P} +  \bm{G} T_{\text{amb}},
\end{equation}

\noindent where $\eta$ is the number of thermal nodes, $\bm{A} \in \mathbb{R}^{\eta\times \eta}$ is a matrix of capacitances, $\bm{B} \in \mathbb{R}^{\eta \times \eta}$ is a matrix of thermal conductances, $\bm{T} \in \mathbb{R}^{\eta \times 1}$ is a vector of temperatures at each node, $\bm{P} \in \mathbb{R}^{\eta \times 1}$ is a vector of power consumption of the nodes, and $\bm{G} \in \mathbb{R}^{\eta \times 1}$ is a vector containing the thermal conductance between each node and the ambient.

When the system reaches a steady state, $\bm{A} \bm{T}^{\prime}$ becomes zero as the temperature remains constant in time. Then, considering a single thermal node only, $\bm{B}$, $\bm{G}$ and $\bm{P}$ become scalars (denoted by $B$, $G$, and $P$) and the whole system reduces to linear relation with respect to $P$:
\begin{equation} \label{eq:tp-linear}
    T = \frac{1}{B} P + \frac{G}{B} T_{\text{amb}},
\end{equation}
where $T$ is the steady-state temperature at the thermal node, and $P$ is the power consumption.

In the rest of the paper, we work with the power model instead of the thermal model, assuming that the final transformation from the average power to the steady-state temperature can be done according to \eqref{eq:tp-linear}.

Note that model \eqref{eq:rc-system} does not take into account the temperature-dependent leakage power, contrary to, e.g., Guo et al.~\cite{2019:Guo}. While this might look like a significant drawback, our results in \cref{sec:experiments} show that even such a model is sufficient for temperature reduction when integrated within the optimization framework. 

\subsection{Power Models} \label{sec:thermal-model-implementation}

Following the general discussion on thermal modeling, we continue with descriptions of specific models. As noted in \Cref{sec:thermal-to-power}, since the average power and the steady-state temperature are linearly related, we implement just the models estimating the average power consumption.

\subsubsection{Empirical Sum-Max Model}

First, we summarize the \emph{sum-max} model (\ModelSumMax{}) that we proposed in \cite{2021:Benedikt}. The model is purely empirical; given a window with the allocated tasks, it predicts the average power consumed during the execution of such a window. 

Specifically, given window $\Win{\WinIdx}$ of length $\WinLen{\WinIdx}$ and tasks allocation $(\TaskSet_{1}^{\WinIdx}, \dots, \TaskSet_{\ResNum}^{\WinIdx})$, where set $\TaskSet_{\ResIdx}^{\WinIdx}$ represents tasks allocated to cluster $\Res{\ResIdx}$ in window $\Win{\WinIdx}$, the \ModelSumMax{} model predicts the average power consumption $\PTotal(\Win{\WinIdx})$ as:

\begin{equation}
    \PTotal(\Win{\WinIdx}) = \sum\limits_{\Res{\ResIdx} \in \ResSet} \sum\limits_{\Task{\TaskIdx} \in \TaskSet_{\ResIdx}^{\WinIdx}} \left( \TaskCoefSlope{\TaskIdx}{\ResIdx} \cdot \frac{\TaskProc{\TaskIdx}{\ResIdx}}{\WinLen{\WinIdx}} \right) +  \max\limits_{\substack{\Res{\ResIdx} \in \ResSet \\ \Task{\TaskIdx} \in \TaskSet_{\ResIdx}^{\WinIdx}}} \TaskCoefOffset{\TaskIdx}{\ResIdx} + \PIdle, \label{eq:max-sum-model}
\end{equation}
where $\PIdle$ is the idle power consumption of the platform, and  $\TaskCoefOffset{\TaskIdx}{\ResIdx}$ and $\TaskCoefSlope{\TaskIdx}{\ResIdx}$ are task-specific coefficients obtained via benchmarking. 
The average power of a schedule consisting of multiple windows is calculated as a weighted average of their individual contributions (the weights correspond to the window lengths).

The model is built upon the assumption that power consumption of $z$ instances of task $\Task{\TaskIdx}$ executed independently and in parallel on $z \in \{1,2,\dots,\ResCap{\ResIdx}\}$ cores of cluster $\Res{\ResIdx}$ can be expressed as $(z \cdot \TaskCoefSlope{\TaskIdx}{\ResIdx} + \TaskCoefOffset{\TaskIdx}{\ResIdx} + \PIdle)$. Coefficients $\TaskCoefSlope{\TaskIdx}{\ResIdx}$ and $\TaskCoefOffset{\TaskIdx}{\ResIdx}$ can be related to dynamic and static power consumption incurred by execution of task $\Task{\TaskIdx}$ on $\Res{\ResIdx}$. The intuition behind choosing the maximum $\TaskCoefOffset{\TaskIdx}{\ResIdx}$ is that different tasks need to power up different number of resources shared between cores (interconnects, shared execution units, etc.) and therefore only the task requesting the most resources determines the static power consumption of these resources.
At the end of this chapter, we present a numerical example illustrating the calculation of the \ModelSumMax{} power model for one specific window.

\paragraph{Platform Characteristic} In the context of the modeling and optimization framework discussed in this paper, the platform characteristic is given by a single parameter only, which is the idle power consumption, $\PIdle$.
\paragraph{Task Characteristics} The sum-max model needs two coefficients for each task and cluster, i.e., task characteristics $\TaskChar{\TaskIdx}$ for a platform with two clusters are represented by a four-tuple for each $\Task{\TaskIdx} \in \TaskSet$, $\TaskChar{\TaskIdx} = (\TaskCoefSlope{\TaskIdx}{1}, \TaskCoefOffset{\TaskIdx}{1}, \TaskCoefSlope{\TaskIdx}{2}, \TaskCoefOffset{\TaskIdx}{2})$.  Following our methodology introduced in~\cite{2021:Benedikt}, these coefficients are determined experimentally, as shown in \cref{sec:platf-task-char}.

\subsubsection{Linear Regression Model}

The sum-max model was designed to estimate the average power consumption of a whole window. Its simple form allows for relatively straightforward integration with the optimization methods, e.g., with Integer Linear Programming, as we have shown in \cite{2021:Benedikt}. However, the model may fail to provide accurate prediction when the tasks of very distinct lengths are present in the window. Specifically, the max term counts with the largest $\TaskCoefOffset{\TaskIdx}{\ResIdx}$ only, which might represent one of the short tasks that possibly ends early in the window. In that case, the predicted power might overestimate the actual one.

To overcome this shortcoming, we designed another power estimation model based on linear regression, which was successfully used in the context of power consumption estimation \cite{2020:Shahid}. Instead of estimating the average power of the whole window, we split the window into several intervals, within which each core either executes a single benchmark for the whole time or remains idle for the whole time. We call them \emph{\PIIntervals{}} intervals. The situation is illustrated in \Cref{fig:splitting-window}. Then, the model estimates the power consumption of each such interval. Similarly to the \ModelSumMax{} model, the overall average power consumption is then estimated by a weighted average of the intervals' individual contributions.

Thanks to the decomposition into the \PIIntervals{} intervals, data acquisition and parameter identification of the model becomes easier since the timing and overlaps of the individual tasks do not need to be considered (each core is either processing or idling for the whole interval). A further advantage is that such a model can be used even if the temporal isolation constraints (windows) are not considered since essentially any multi-core schedule can be divided into such intervals (simply by projecting the start times and end times of tasks to the time axis; the intervals are then defined by every two consecutive projected time points). On the contrary, the integration of the \ModelLR{} model with the optimization might be harder since the lengths of the \PIIntervals{} intervals are not known a priori as they depend on the allocation of the tasks to windows. 

\begin{figure}
    \centering
\newcommand{\TaskGantt}[5]{
	\draw[fill=#5] ($(#2,-#1) - (0,0.2)$) rectangle node[font=\scriptsize,pos=0.5] {#4} ($(#3, -#1) - (0,0.8)$);
}

\begin{tikzpicture}[yscale=0.6,xscale=1.4]

\TaskGantt{0}{0}{4.5}{Task 1 (a2time-4K)}{Clr1!40!white};
\TaskGantt{2}{0}{5.5}{Task 2 (canrdr-4M)}{Clr2!40};
\TaskGantt{4}{0}{7}{Task 3 (membench-1M-RO-S)}{Clr3!40};

\foreach \i in {0, 4.5, 5.5, 7}{
	\draw[dashed] (\i, 0) -- (\i, -7);
}

\draw[latex-latex, thick] (0,-6.5) -- node[anchor=north,font=\tiny] {Interval 1} (4.5,-6.5);
\draw[latex-latex, thick] (4.5,-6.5) -- node[anchor=north,font=\tiny] {Interval 2} (5.5,-6.5);
\draw[latex-latex, thick] (5.5,-6.5) -- node[anchor=north,font=\tiny] {Interval 3} (7,-6.5);

\foreach \i in {1,2,3,5}{
	\draw[thick, densely dotted] (0,-\i) -- (7,-\i);
}

\draw[thick] (0,-4) -- (7,-4);

\draw[thick] (0,0) rectangle (7,-6);
\node[anchor=south] at (3.5, 0) {\footnotesize Window $W$};

\foreach \i in {1,...,6}{
	\node[anchor=west, font=\footnotesize, xshift=5pt] at  ($(7, -\i) + (0,0.5)$) {Core \i};
}

\begin{scope}[xshift=-5pt]
	\draw[latex-latex, thick] (0,-4) -- node[sloped, anchor=south, pos=0.5] {\scriptsize little cluster} (0,0);
	\draw[latex-latex, thick] (0,-6) -- node[sloped, anchor=south, pos=0.5] {\scriptsize big cluster} (0,-4);
\end{scope}

\end{tikzpicture}
    \caption{Illustration of \PIIntervals{} intervals needed for the \ModelLR{} model on a platform with two clusters named \emph{big} and \emph{little}.}
    \label{fig:splitting-window}
\end{figure}

To describe the Linear Regression model (\ModelLR{}), let us assume that we have some interval $I$ with allocated tasks. By $i(\ResIdx, \CoreIdx)$ we denote the index of the task, which is allocated in interval $I$ to cluster $\Res{\ResIdx} \in \ResSet$ and its core $\CoreIdx \in \{1,\dots, \ResCap{\ResIdx}\}$. If the core is idle, we assume $i(\ResIdx, \CoreIdx) = 0$ (where $\Task{0}$ represents an idle task). Now, let us assume that behavior of task $\Task{\TaskIdx}$ on each cluster $\Res{\ResIdx}$ is characterized by a vector of real numbers $\bm{\hat{x}}_{\TaskIdx, \ResIdx}$ (e.g., $\bm{\hat{x}}_{\TaskIdx, \ResIdx} = (\TaskCoefOffset{\TaskIdx}{\ResIdx}, \TaskCoefSlope{\TaskIdx}{\ResIdx})$). We assume that idle task $\Task{0}$ is characterized by zero vector, $\bm{\hat{x}}_{0, \ResIdx} = \bm{0} \ \forall \Res{\ResIdx} \in \ResSet$. Then, the average power consumption of interval $I$ can be estimated by \ModelLR{} model as:

\begin{equation} \label{eq:lr-general}
    P(I) = 
    \sum\limits_{\Res{\ResIdx} \in \ResSet}
    \sum\limits_{\CoreIdx=1}^{\ResCap{\ResIdx}} 
    \left( 
        \bm{\hat{x}}_{i(\ResIdx,\CoreIdx),\ResIdx} \circ \bm{\beta}_{\ResIdx, \CoreIdx}
    \right) + \PIdle,
\end{equation}
where $\circ$ is the scalar product operator and $\PIdle$ is the constant intercept term. Note that when no task is executed (all are idle), the prediction is exactly $\PIdle$, which is the platform's idle power consumption.
In this general form, the regression coefficients $\bm{\beta}_{\ResIdx, \CoreIdx}$ are possibly different for each core. However, we assume that all cores of each cluster are identical, so arbitrary permutation of tasks allocated to a single cluster should lead to the same power consumption. By this, we can simplify the model \eqref{eq:lr-general} to:

\begin{equation}  \label{eq:lr-specific}
\begin{aligned}
    P(I) & = 
    \sum\limits_{\Res{\ResIdx} \in \ResSet}
    \sum\limits_{\CoreIdx=1}^{\ResCap{\ResIdx}} 
    \left( 
        \bm{\hat{x}}_{i(\ResIdx,\CoreIdx),\ResIdx} \circ \bm{\beta}_{\ResIdx}
    \right) + \PIdle \\
     & = 
    \sum\limits_{\Res{\ResIdx} \in \ResSet}
    \bm{\beta}_{\ResIdx} \circ \left(
    \sum\limits_{\CoreIdx=1}^{\ResCap{\ResIdx}} 
        \bm{\hat{x}}_{i(\ResIdx,\CoreIdx),\ResIdx}
    \right) + \PIdle,
\end{aligned}
\end{equation}
where all cores of the single cluster have the same regression coefficients. 
For simpler understanding of the calculation of the \ModelLR{} power model, we present a numerical example at the end of this chapter.

\paragraph{Platform and Tasks Characteristics} Similarly to \ModelSumMax{} model, the \ModelLR{} model needs just the idle power $\PIdle$ as the platform characteristic. On the other hand, the individual tasks characteristics now correspond to the elements of input vector $\bm{\hat{x}}_{\TaskIdx,\ResIdx}$. For simplicity and better comparison, we can use already identified coefficients $\TaskCoefOffset{\TaskIdx}{\ResIdx}$ and $\TaskCoefSlope{\TaskIdx}{\ResIdx}$, which are also the only tasks characteristics used by \ModelSumMax{} model. In general, we could also include more characteristics obtained, e.g., by monitoring the performance counters during the tasks execution \cite{2021:Sahid}.

The regression coefficients $\bm{\beta}_{\ResIdx}$ are obtained experimentally as shown in \cref{sec:ident-regr-coeff}.

\begin{example}
  To illustrate the calculation of both power models numerically, let us assume three tasks on a platform with two clusters named \emph{big} and \emph{little}, as illustrated in \cref{fig:splitting-window}. Each task executes different code (called kernel) described in detail later in \cref{sec:benchmarking-kernels}. The first task executes \emph{a2time-4K} kernel for \SI{450}{\milli\second} at the little cluster ($\ResIdx = 1$). The second task, also assigned to the little cluster, executes \emph{canrdr-4M} kernel for \SI{550}{\milli\second}. Finally, the third task, which is assigned to the big cluster ($\ResIdx = 2$), executes \emph{membench-1M-RO-S} kernel for \SI{700}{\milli\second}. Hence, the length of the window is \SI{700}{\milli\second}. Assume that the platform characteristic $P_{\text{idle}} = \SI{5.5}{\watt}$ and relevant task characteristics coefficients are\footnote{The values correspond to the I.MX8~MEK platform -- see \cref{sec:platf-task-char} and \ref{app:kernel-chars-values}.}: $\TaskCoefSlope{1}{1} = 0.25$, $\TaskCoefOffset{1}{1} = 0.25$, $\TaskCoefSlope{2}{1} = 0.41$, $\TaskCoefOffset{2}{1} = 1.36$, $\TaskCoefSlope{3}{2} = 1.24$, and $\TaskCoefOffset{3}{2} = 1.22$.
The \ModelSumMax{} power model just takes the maximum offset coefficient (1.36) and adds the activity coefficients contributions of the individual tasks, i.e., the estimated power consumption is:
  $P_{\text{idle}} + 1.36 + \left( 0.25 \cdot \frac{450}{700} + 0.41 \cdot \frac{550}{700} + 1.24 \cdot \frac{700}{700} \right) \doteq \SI{8.58}{\watt}.$
To evaluate the \ModelLR{} model, the whole window is split to three \PIIntervals{} intervals. The first interval is \SI{450}{\milli\second} long and all three tasks are executed during it. The second interval is \SI{100}{\milli\second} long and covers the second and third task. The last interval is \SI{150}{\milli\second} long, and covers only the third task. The power is estimated individually for each interval, and averaged out in the end. Assume the linear regression coefficients are (see \cref{tab:regression-coefficients}): $\beta_{1,1}=1.205$, $\beta_{1,2}=0.969$, $\beta_{2,1}=0.270$, $\beta_{2,2}=0.456$. 
For the first interval, the prediction is $ \left[\beta_{1,1}(\TaskCoefSlope{1}{1} + \TaskCoefSlope{2}{1}) + \beta_{1,2} \cdot\TaskCoefSlope{3}{2}\right] + \left[\beta_{2,1}(\TaskCoefOffset{1}{1} + \TaskCoefOffset{2}{1}) + \beta_{2,2} \cdot\TaskCoefOffset{3}{2}\right]$, which is $[1.205 \cdot (0.25+0.41) + 0.969 \cdot 1.24] + [0.270 \cdot (0.25+1.36) + 0.456 \cdot 1.22] \doteq 2.99$. Since this interval lasts only \SI{450}{\milli\second}, this number is then averaged to $2.99 \cdot \frac{450}{700} \doteq \SI{1.92}{\watt}$. The calculations for the other two intervals are analogous; we report just their averaged contributions, which are $\SI{0.37}{\watt}$ and $\SI{0.38}{\watt}$, respectively. In total, the power predicted by \ModelLR{} power model is $P_{\text{idle}} + 1.92 + 0.37+0.38 = \SI{8.17}{\watt}$. 

\end{example}

\section{Optimization Methods} \label{sec:solution-methods}

In \Cref{sec:thermal-modeling}, we revised the thermal modeling, and summarized two specific power models, namely \ModelSumMax{} and \ModelLR{} models. Further, we defined all the necessary platform and task characteristics. In this section, we develop optimization methods that incorporate the power models and solve the problem of task allocation while minimizing the estimated power consumption. The resulting combination presents a novel solution of the mentioned problem.

First, \cref{sec:method-ilp} summarizes the optimizer based on Integer Linear Programming (ILP) and the \ModelSumMax{} model that we proposed in \cite{2021:Benedikt}. 
Second, we discuss the optimization based on the \ModelLR{} power model in \cref{sec:method-lr-general}. Finally, we introduce an informed greedy heuristic and uninformed idle-time optimizers used for further comparison in \cref{sec:method-ref,sec:idle-optimizer}, respectively. 

\subsection{ILP and Sum-Max Model} \label{sec:method-ilp}

The original implementation that we proposed in \cite{2021:Benedikt} relies on a simple encoding of the scheduling problem to the ILP formalism. Binary variable $\VarAss{\TaskIdx}{\WinIdx}{\ResIdx}$ equals to 1 if and only if task $\Task{\TaskIdx} \in \TaskSet$ is allocated to window $\Win{\WinIdx} \in \WinSet$ and cluster $\Res{\ResIdx} \in \ResSet$.  Then, all the resource constraints can be simply written as:

\begin{align}
    & \sum\limits_{\Win{\WinIdx} \in \WinSet} \sum\limits_{\Res{\ResIdx} \in \ResSet} \VarAss{\TaskIdx}{\WinIdx}{\ResIdx} = 1 \quad \forall \Task{\TaskIdx} \in \TaskSet, \label{eq:model-assignment}\\
    & \sum\limits_{\Task{\TaskIdx} \in \TaskSet} \VarAss{\TaskIdx}{\WinIdx}{\ResIdx} \leq \ResCap{\ResIdx} \quad \forall {\Win{\WinIdx} \in \WinSet}, \Res{\ResIdx} \in \ResSet, \label{eq:model-capacity}
\end{align}
meaning that each task is assigned to some cluster and core and that the capacity of each cluster is respected. To model the length of the individual windows, variable $\WinLen{\WinIdx}$ is introduced for each window $\Win{\WinIdx} \in \WinSet$. The length is then linked to the assignment variables by

\begin{equation}
    \WinLen{\WinIdx} \geq \VarAss{\TaskIdx}{\WinIdx}{\ResIdx} \cdot \TaskProc{\TaskIdx}{\ResIdx} \quad  \forall \Task{\TaskIdx} \in \TaskSet, \Win{\WinIdx} \in \WinSet, \Res{\ResIdx} \in \ResSet, \label{eq:model-len-max-partition}
\end{equation}
and constrained by the major frame length $\MF$:

\begin{equation}
    \sum\limits_{\Win{\WinIdx} \in \WinSet} \WinLen{\WinIdx} \leq \MF. \label{eq:model-MF}
\end{equation}

Finally, the \ModelSumMax{} model \eqref{eq:max-sum-model} is linearized to fit the ILP formalism and rewritten to the objective function~\eqref{eq:model-objective}. The idle power $\PIdle$ is not included in the objective since it is considered to be constant. The power consumption predictions are averaged over all windows with the weights corresponding to the windows lengths. The non-linear \emph{max} term is (in each window $\Win{\WinIdx} \in \WinSet$) replaced by continuous variable $y_{\WinIdx}$, which serves as its upper bound. When the solver reaches the optimum, this upper bound becomes tight. The link between $y_{\WinIdx}$ and $\VarAss{\TaskIdx}{\WinIdx}{\ResIdx}$ is formed by \emph{big-M}, where $M$ is a sufficiently large constant. The objective~\eqref{eq:model-objective} represents the estimated average power consumption. The whole ILP model encoding the scheduling problem and integrating the \ModelSumMax{} model then becomes:

\begin{align}
    \text{\ModelILPSMOrig:} \quad & \min \frac{1}{\MF} 
    \sum\limits_{\Win{\WinIdx} \in \WinSet} 
        \left( \sum\limits_{\Task{\TaskIdx} \in \TaskSet} 
        \sum\limits_{\Res{\ResIdx} \in \ResSet} 
        (\VarAss{\TaskIdx}{\WinIdx}{\ResIdx} \cdot \TaskCoefSlope{\TaskIdx}{\ResIdx} \cdot \TaskProc{\TaskIdx}{\ResIdx}) 
        + y_{\WinIdx} \right) 
    \label{eq:model-objective} \\
    & \text{subject to:} \nonumber \\
    & y_{\WinIdx} \geq \TaskCoefOffset{\TaskIdx}{\ResIdx} \cdot \WinLen{\WinIdx} - M \cdot (1-\VarAss{\TaskIdx}{\WinIdx}{\ResIdx}) \quad \forall \Task{\TaskIdx} \in \TaskSet, \Win{\WinIdx}  \in \WinSet, \Res{\ResIdx} \in \ResSet, \label{eq:model-y-ub} \\
    & \text{\eqref{eq:model-assignment}, \eqref{eq:model-capacity}, \eqref{eq:model-len-max-partition}, \eqref{eq:model-MF}.} \nonumber \\
    & \VarAss{\TaskIdx}{\WinIdx}{\ResIdx} \in \{0,1\}, \WinLen{\WinIdx} \in \SetIntNonNeg, y_{\WinIdx} \in \SetRealNonNeg \nonumber
\end{align}

\subsection{Optimization Based on the Linear Regression Model} \label{sec:method-lr-general}

Contrary to the \ModelSumMax{} model, which predicts the power for each window, the \ModelLR{} model predicts the power for \PIIntervals{} intervals only.
Direct integration within the ILP formalism proved to be quite laborious\footnote{We implemented the model, but it became rather complex, mainly due to the necessary linearization, and its performance was poor even for small instances (we failed to obtain reasonable upper and lower bounds for instance with 20 tasks even after one day of solving by state-of-the-art solver Gurobi~\cite{gurobi}).}; therefore, we followed two different paths. 
First, we neglect the \PIIntervals{} intervals and assume that each task executes for the whole duration of the window to which it is allocated. We call this simplified model as \texttt{LR-UB}. Then we can simply formulate the optimization as a quadratic programming optimization problem as shown in \cref{sec:lr-and-QP}. 
Second, we use a different optimization framework, namely the black-box optimization based on a genetic algorithm, which can simply integrate the \ModelLR{} model as a part of the fitness evaluation. We describe this approach in \cref{sec:bb-lr-model}.

\subsubsection{Integration of LR to QP} \label{sec:lr-and-QP}

We assume that all tasks allocated to one window are executed for the whole length of the window. Therefore, the execution times of the individual tasks are assumed to be potentially longer than they are in reality. In consequence, \PIIntervals{} intervals are completely neglected (each window becomes a single \PIIntervals{} interval). In \cref{fig:intervals-relaxed}, we illustrate the assumed extensions of the task execution times in gray (the original \PIIntervals{} intervals are shown in \cref{fig:splitting-window}).
In a sense, the idea is similar to the \emph{max} term in the \ModelSumMax{} model. We hope that by minimizing the upper bound instead of the original objective, we get a schedule that performs reasonably well in practice while keeping the formulation relatively simple. 

\begin{figure}
    \centering
\newcommand{\TaskGantt}[5]{
	\draw[fill=#5] ($(#2,-#1) - (0,0.2)$) rectangle node[font=\tiny,pos=0.5] {#4} ($(#3, -#1) - (0,0.8)$);
}

\begin{tikzpicture}[yscale=0.6,xscale=1.4]

\TaskGantt{0}{0}{4.5}{Task 1}{Clr1!40!white}; \TaskGantt{0}{4.5}{7}{Task 1 (extended)}{lightgray};
\TaskGantt{2}{0}{5.5}{Task 2}{Clr2!40}; \TaskGantt{2}{5.5}{7}{Task 2 (extended)}{lightgray};
\TaskGantt{4}{0}{7}{Task 3}{Clr3!40};

\draw[latex-latex, thick] (0,-6.5) -- node[anchor=north,font=\tiny] {Interval 1} (7,-6.5);

\foreach \i in {1,2,3,5}{
	\draw[thick, densely dotted] (0,-\i) -- (7,-\i);
}

\draw[thick] (0,-4) -- (7,-4);

\draw[thick] (0,0) rectangle (7,-6);
\node[anchor=south] at (3.5, 0) {\footnotesize Window $W$};

\foreach \i in {1,...,6}{
	\node[anchor=west, font=\footnotesize, xshift=5pt] at  ($(7, -\i) + (0,0.5)$) {Core \i};
}

\begin{scope}[xshift=-5pt]
	\draw[latex-latex, thick] (0,-4) -- node[sloped, anchor=south, pos=0.5] {\scriptsize little cluster} (0,0);
	\draw[latex-latex, thick] (0,-6) -- node[sloped, anchor=south, pos=0.5] {\scriptsize big cluster} (0,-4);
\end{scope}

\end{tikzpicture}
    \caption{Illustration of a simplified window with a single \PIIntervals{} interval only, as used by \texttt{LR-UB}.}
    \label{fig:intervals-relaxed}
\end{figure}

Then, we follow the same steps as for the integration of the \ModelSumMax{} model with the ILP. We use the same binary variables $\VarAss{\TaskIdx}{\WinIdx}{\ResIdx}$ deciding whether task $\Task{\TaskIdx} \in \TaskSet$ is allocated to window $\Win{\WinIdx} \in \WinSet$ and cluster $\Res{\ResIdx} \in \ResSet$. Constraints modeling the task allocation and the resource capacity and limiting the major frame length are the same as before, only the objective changes, as the power is now predicted by the \ModelLR{} model for each window. The whole model is now as follows:

\begin{align}
    \text{\ModelQPLR:} \quad & \min\frac{1}{\MF} \cdot \sum\limits_{\Win{\WinIdx} \in \WinSet}
    \WinLen{\WinIdx} \cdot \sum\limits_{\Task{\TaskIdx} \in \TaskSet} \sum\limits_{\Res{\ResIdx} \in \ResSet} \VarAss{\TaskIdx}{\WinIdx}{\ResIdx} \cdot \underbrace{\left( \TaskCoefSlope{\TaskIdx}{\ResIdx} \cdot \beta_{1,\ResIdx} + \TaskCoefOffset{\TaskIdx}{\ResIdx} \cdot \beta_{2,\ResIdx} \right)}_{\star}
    \label{eq:model-lr-objective} \\
    & \text{subject to:} \nonumber \\
    & \text{\eqref{eq:model-assignment}, \eqref{eq:model-capacity}, \eqref{eq:model-len-max-partition}, \eqref{eq:model-MF}.} \nonumber \\
    & \VarAss{\TaskIdx}{\WinIdx}{\ResIdx} \in \{0,1\}, \WinLen{\WinIdx} \in \SetIntNonNeg \nonumber
\end{align}

Clearly, the objective becomes quadratic (due to multiplication of $\VarAss{\TaskIdx}{\WinIdx}{\ResIdx}$ and $\WinLen{\WinIdx}$).\footnote{We use Gurobi solver for both Linear and Quadratic optimization.} Note that for all the tested benchmarking kernels (see \ref{app:kernel-chars-values}) and identified linear regression coefficients (see \cref{tab:regression-coefficients}), the expression denoted by $\star$ in \eqref{eq:model-lr-objective} is positive. Furthermore, $\star$ is zero for the idle task.
Therefore, the objective~\eqref{eq:model-lr-objective} becomes an upper bound on the original \ModelLR{} value (by the original \ModelLR{} we mean the \ModelLR{} model, which does not neglect the \PIIntervals{} intervals). 

\subsubsection{Black-Box Optimizer} \label{sec:bb-lr-model}

Instead of using Mathematical Programming and building a complicated model, we can find the solution using a conceptually different black-box optimization framework.
The objective function is not given in a closed form, but only its outputs can be observed provided the inputs. For us, given the full tasks allocation (schedule), we can compute the average power consumption based on the \ModelLR{} model (or possibly any other model). The black-box optimization algorithm searches through the space of all allocations for some solutions that are probably better than other feasible solutions and choose the best one among them w.r.t the given fitness function (here the \ModelLR{} model). Note that the search does not enumerate all possible solutions and as such, it may not find the global optimum.

There are many algorithms that can be used for black-box optimization, including, for example, Particle Swarm Optimization, Differential Evolution, or Genetic Algorithms. Some of these algorithms are already implemented in existing libraries, which are often optimized for speed, easily usable, and open-source. 
We use \emph{Genetic Algorithm} (\AlgGA{}) from \emph{Evolutionary} package implemented in Julia \cite{2022:Julia-Evolutionary}. We use standard two-point crossover and BGA mutation; mutation and crossover rates are set to 0.2 and 0.8, respectively. The selection is done according to a uniform ranking scheme (discarding the lowest \SI{10}{\percent} of the population), and the population size is set to $50 \cdot \vert \TaskSet \vert $.

We represent the position of each task $\Task{\TaskIdx} \in \TaskSet$ in the schedule by continuous variable $x_{\TaskIdx} \in [0,1)$. In order to optimize the allocation problem using the continuous variables, we introduce the following transformation: Each variable $x_{\TaskIdx}$ is evenly split to $\ResNum$ intervals, i.e., for two clusters, we obtain interval $[0,0.5)$ representing the allocation to the first cluster and interval $[0.5, 1)$ representing the allocation to the second one. Each such sub-interval is then again evenly split to $\WinNum$ intervals representing the allocation to the individual windows.

Still, it might happen that allocation represented by variables $x_i$ would be infeasible -- either due to the major frame length (when allocated windows are too long) or due to resource capacity constraints (when too many tasks are allocated to the same window and resource). There are several ways to handle this issue. One option is to use such a black-box solver that supports constrained optimization (\AlgGA{} can do that). Another option is to introduce post-processing that would try to reconstruct some feasible solution from the infeasible assignment. Even though it would appear that solely the former option solves the problem, too many infeasible solutions slow down the convergence of the optimization algorithm. Therefore, we use both presented options -- the former for the major frame length constraint and the latter for the resource constraints.

The post-processing (reconstruction) procedure is described by Algorithm~\ref{alg:reconstruction} in \ref{app:reconstruction-algorithm}. Informally, the preferred allocation of the tasks is pre-computed based on the transformation described above. Then, starting from the first window, the allocation of the tasks is iteratively fixed. If the task cannot be added to the current window (i.e., the resource capacity would be exceeded), its preferred allocation is moved to the next window (in a cyclic manner). The iteration over all windows is repeated twice. If there are still some unassigned tasks or the major frame length is exceeded, the solution is discarded; otherwise, feasible allocation of the tasks to windows and clusters is returned.

The black-box optimizer iterates over many possible instantiations of $x_{\TaskIdx}$. Every time some instantiation is tested, the schedule (defined by allocations created by Algorithm~\ref{alg:reconstruction}) is reconstructed and evaluated by the \ModelLR{} power model. After a termination condition is met (e.g., time limit or iteration limit is exceeded), the best-so-far solution is returned. 
We execute the algorithm with a pre-defined time limit; if it terminates sooner, it is restarted from a random instantiation of $x_{\TaskIdx}$ until exhausting the time limit.

\subsection{Greedy Heuristic} \label{sec:method-ref}

As a reference method, we describe a greedy heuristic. Such heuristics are often used in on-line real-time scheduling algorithms due to their low computation demand. Contrary to all the previous methods based on ILP, QP, or black-box optimization, the greedy heuristic does not try to search through the whole optimization space (here, set of all possible allocations). Instead, the search space is intentionally restricted in order to decrease the computation time and improve the scalability. 

The heuristic that we present is based on the works of Zhou et al. \cite{2016:Zhou} and Kuo et al. \cite{2015:Kuo}, but its main idea is rather general and applicable in the wider context. The tasks are sorted by their energy consumption and processed one by one in a non-increasing order (the most energy-consuming task goes first). In each iteration, the currently processed task is assigned to the cheapest computing cluster (w.r.t. energy consumption). The assignment is done only if some feasible schedule exists even for all the remaining (still unprocessed) tasks, i.e., the assignment cannot be fixed if it would cause infeasibility.

For the tasks ordering, we use analogous methodology that is used in \cite{2016:Zhou} (in Algorithm 1) -- we can identify the parameter $\mu_{\TaskIdx}$ used in \cite{2016:Zhou} with task characteristic $\TaskCoefSlope{\TaskIdx}{\ResIdx}$ since both of these parameters represent tasks dynamic power consumption to some extent. Then, the task $\Task{\TaskIdx}$ is assigned to cluster $\Res{\ResIdx} \in \ResSet$ that minimizes $\TaskCoefSlope{\TaskIdx}{\ResIdx} \cdot \TaskProc{\TaskIdx}{\ResIdx}$ (i.e., expected task energy consumption). Before each assignment, feasibility needs to be checked. When considering the windows, it becomes a bit tricky because these windows make the scheduling on the individual clusters and their cores dependent on each other (without the windows, the situation is much simpler since only the utilization bound needs to be checked). 

To check the feasibility, we use a modified ILP model as presented in \Cref{sec:method-ilp}:

\begin{align}
    \text{\ModelILPFeasibility:} \quad & \min 0 \\
    & \text{subject to:} \nonumber \\
    &  \sum_{\Win{\WinIdx} \in \WinSet} \VarAss{\TaskIdx}{\WinIdx}{r(\Task{\TaskIdx})} = 1\quad \forall \Task{\TaskIdx} \in \mathcal{T}_{\text{fixed}} \label{eq:heur-fixed-assignment} \\
    & \text{\eqref{eq:model-assignment}, \eqref{eq:model-capacity}, \eqref{eq:model-len-max-partition}, \eqref{eq:model-MF}}, \nonumber\\
    & \VarAss{\TaskIdx}{\WinIdx}{\ResIdx} \in \{0,1\}, \WinLen{\WinIdx} \in \SetIntNonNeg \nonumber
\end{align}
where $\mathcal{T}_{\text{fixed}}$ represents the set of tasks with already fixed assignment and $r: \TaskSet \rightarrow \{1,\dots, \ResNum\}$ maps the tasks with fixed assignment to the index of their assigned cluster. 
The whole greedy heuristic is summarized in Algorithm~\ref{alg:greedy}.

Note that solving \ModelILPFeasibility{} model is easier compared to \ModelILPSMOrig{} as the solver can terminate after finding any feasible solution, whereas in the latter case, it needs to explore the whole search space somehow (iterating over multiple feasible solutions).

\begin{algorithm}
\SetKwInOut{Input}{input}
\SetKwInOut{Output}{output}
\SetKwProg{Fun}{Function}{ is}{end}
\Input{set of tasks $\TaskSet$, set of clusters $\ResSet$, major frame length $\MF$}
\Output{assignment of the tasks to resources}

\Fun{CheckFeasibility($\TaskSet_{\text{fixed}} \subseteq \TaskSet, r: \TaskSet_{\text{fixed}} \rightarrow \{ 1, \dots, \ResNum \}$)}{
    \If{a feasible solution to \ModelILPFeasibility{} with fixed tasks assignment of $\mathcal T_{\text{fixed}}$ given by $r$ exists}{\Return{true}}
    \lElse{\Return{false}}
}

sort all tasks $\Task{\TaskIdx} \in \TaskSet$ by $\max_{\Res{\ResIdx} \in \ResSet} \{ \TaskCoefSlope{\TaskIdx}{\ResIdx} \cdot \TaskProc{\TaskIdx}{\ResIdx} \}$ in non-increasing order

$\mathcal T_{\text{fixed}} = \{ \}$

\ForEach{task $\Task{\TaskIdx} \in \mathcal T$}{
    
    $\ResSet_{\text{sorted}} \gets {}$ sort $\ResSet$ by  $\max_{\Res{\ResIdx} \in \ResSet} \{ \TaskCoefSlope{\TaskIdx}{\ResIdx} \cdot \TaskProc{\TaskIdx}{\ResIdx} \}$ in non-decreasing order
    
    \ForEach{cluster $\Res{\ResIdx} \in \ResSet_{\text{sorted}}$}{
        assign $\Task{\TaskIdx}$ to $\Res{\ResIdx}$, $r(\Task{\TaskIdx}) = \ResIdx$
    
        \If{ CheckFeasibility($\TaskSet_{\text{fixed}} \cup \{ \Task{\TaskIdx} \}, r$) }{
            $\TaskSet_{\text{fixed}} \gets \TaskSet_{\text{fixed}} \cup \{ \Task{\TaskIdx} \}$
            
            break
        }
    }

}

 \If{$\TaskSet_{\text{fixed}} = \TaskSet$}{
    \Return{assignment of tasks to clusters and windows given by the solution of \ModelILPFeasibility{} with fixed cluster assignments defined by $r$}
 }
 \Else{{error: feasible assignment of tasks to resources does not exist}}

 \caption{Greedy heuristic.}
 \label{alg:greedy} 
\end{algorithm}

\subsection{Optimizer Minimizing/Maximizing the Idle Time} \label{sec:idle-optimizer}

Finally, we present two more optimizers, this time uninformed, i.e., not using any task or platform characteristics.
These methods simply optimize the idle (i.e., non-processing) time within the major frame. We present them as ILP models.

First, the total processing time $t_{\text{processing}}$, can be expressed in terms of variables $\VarAss{\TaskIdx}{\WinIdx}{\ResIdx}$ introduced in \ModelILPSMOrig{} as follows:

\begin{equation} \label{eq:t-proc-heur}
    t_{\text{processing}} = \sum\limits_{\Task{\TaskIdx} \in \TaskSet} \sum\limits_{\Win{\WinIdx} \in \WinSet} \sum\limits_{\Res{\ResIdx} \in \ResSet} \VarAss{\TaskIdx}{\WinIdx}{\ResIdx} \cdot \TaskProc{\TaskIdx}{\ResIdx}.
\end{equation}

Then the total idle time $t_{\text{idle}}$ within the major frame is 

\begin{equation} \label{eq:t-idle-heur}
    t_{\text{idle}} = \left(\MF \cdot \sum\limits_{\Res{\ResIdx} \in \ResSet} \ResCap{\ResIdx}\right) - t_{\text{processing}}.
\end{equation}

Now, the first model, \ModelILPIdleMax{}, maximizes the idle time in the hope that long idle periods allow for the platform to cool down. Also, a schedule with maximal idle time is beneficial from the perspective of the practitioner. Sometimes, the instance changes and some more tasks need to be added for the execution; in such a case, schedules with long idle periods offer the space to do so. The model can be formalized as:

\begin{align}
    \text{\ModelILPIdleMax:} \quad & \max t_{\text{idle}} \\
    & \text{subject to:} \nonumber \\
    & \text{\eqref{eq:model-assignment}, \eqref{eq:model-capacity}, \eqref{eq:model-len-max-partition}, \eqref{eq:model-MF}, \eqref{eq:t-proc-heur}, \eqref{eq:t-idle-heur}.} \nonumber \\
    & \VarAss{\TaskIdx}{\WinIdx}{\ResIdx} \in \{0,1\}, \WinLen{\WinIdx} \in \SetIntNonNeg,  t_{\text{processing}}\in \SetRealNonNeg ,  t_{\text{idle}}\in \SetRealNonNeg \nonumber
\end{align}

Contrary to that, the second model, \ModelILPIdleMin{} minimizes the idle time. The idea is that longer execution time is typically associated with the little cluster (see \Cref{fig:ips-ratio}), which might is more power-efficient. The model is described as

\begin{align}
    \text{\ModelILPIdleMin:} \quad & \min t_{\text{idle}} \\
    & \text{subject to:} \nonumber \\
    & \text{\eqref{eq:model-assignment}, \eqref{eq:model-capacity}, \eqref{eq:model-len-max-partition}, \eqref{eq:model-MF}, \eqref{eq:t-proc-heur}, \eqref{eq:t-idle-heur}.} \nonumber \\
    & \VarAss{\TaskIdx}{\WinIdx}{\ResIdx} \in \{0,1\}, \WinLen{\WinIdx} \in \SetIntNonNeg,  t_{\text{processing}}\in \SetRealNonNeg ,  t_{\text{idle}}\in \SetRealNonNeg \nonumber
\end{align}

\section{Hardware Platforms and Benchmarks} \label{sec:hw-and-benchmarks}

As we discussed previously, we opt for an experimental evaluation of the proposed methods instead of a simulation.
To make the comparison of optimization methods and power models more representative, we conduct the benchmarking and evaluation phases on three platforms, which are briefly described in \Cref{sec:physical-hardware}. Further, we describe the benchmarking kernels selected for the experiments in \Cref{sec:benchmarking-kernels}. 

\subsection{Physical Hardware for Evaluation} \label{sec:physical-hardware}

\begin{figure}
    \centering
    
    \begin{tabular}{ccc}
        \includegraphics[width=0.3\textwidth]{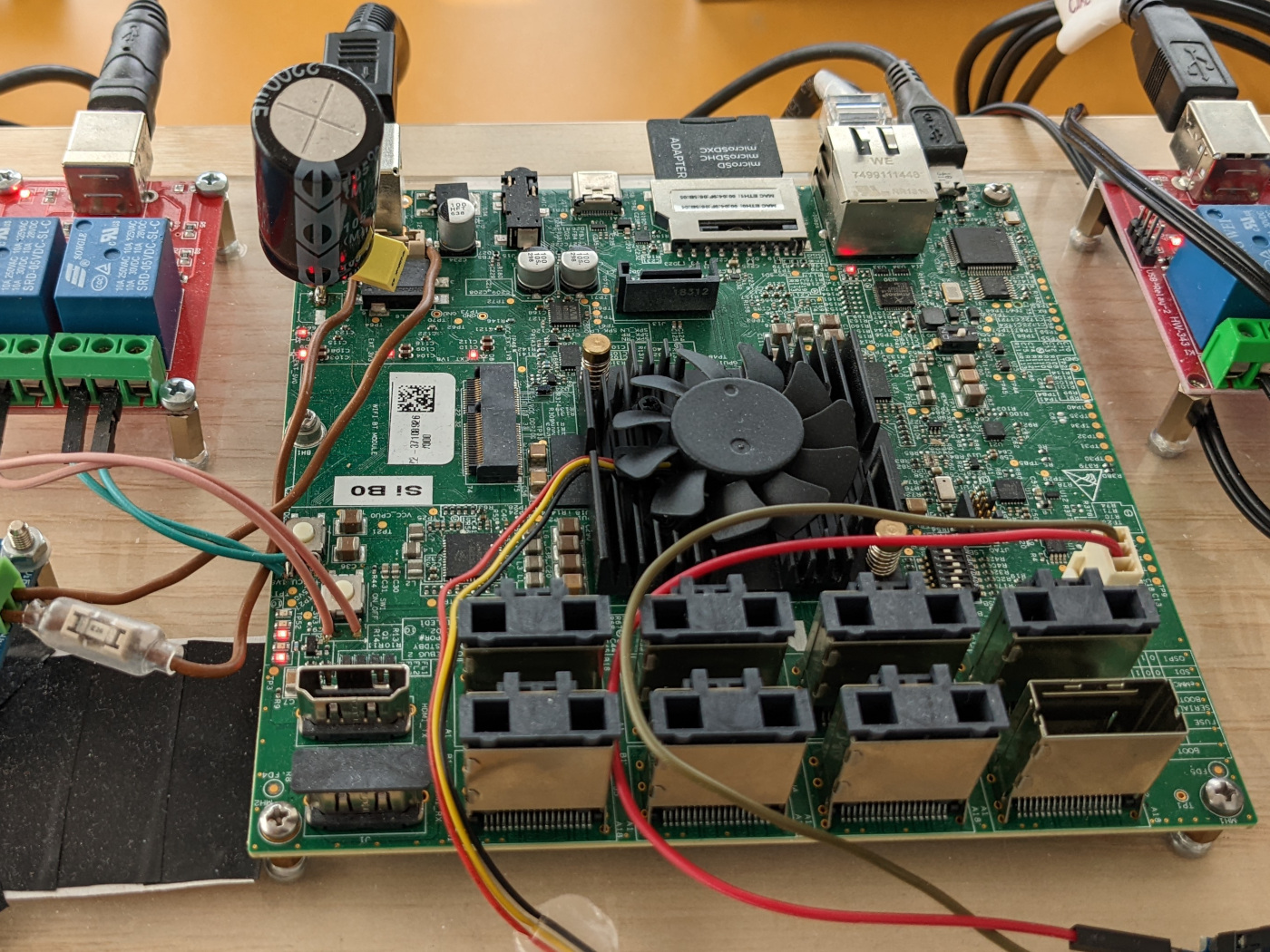} & 
        \includegraphics[width=0.3\textwidth]{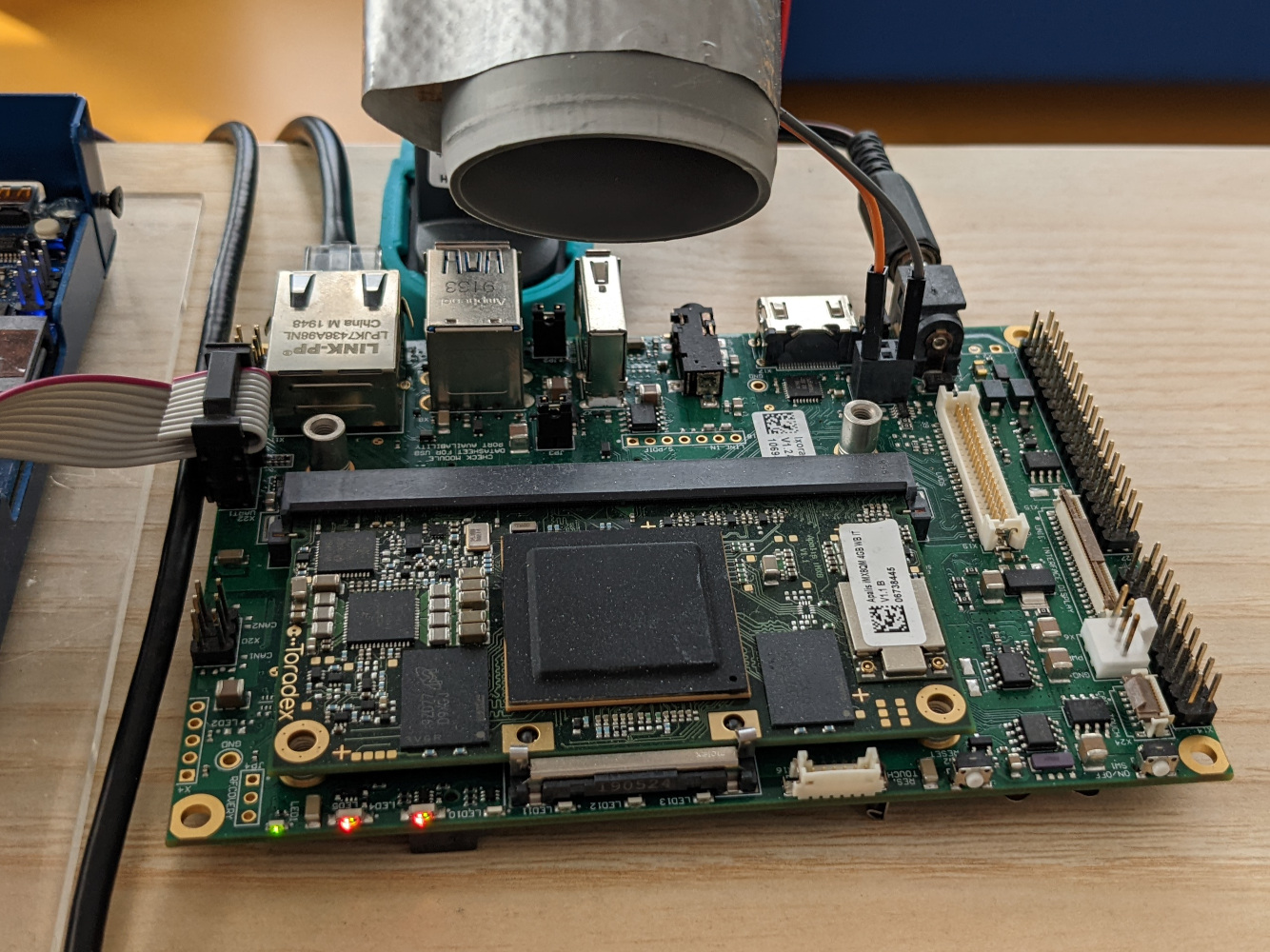} &
        \includegraphics[width=0.3\textwidth]{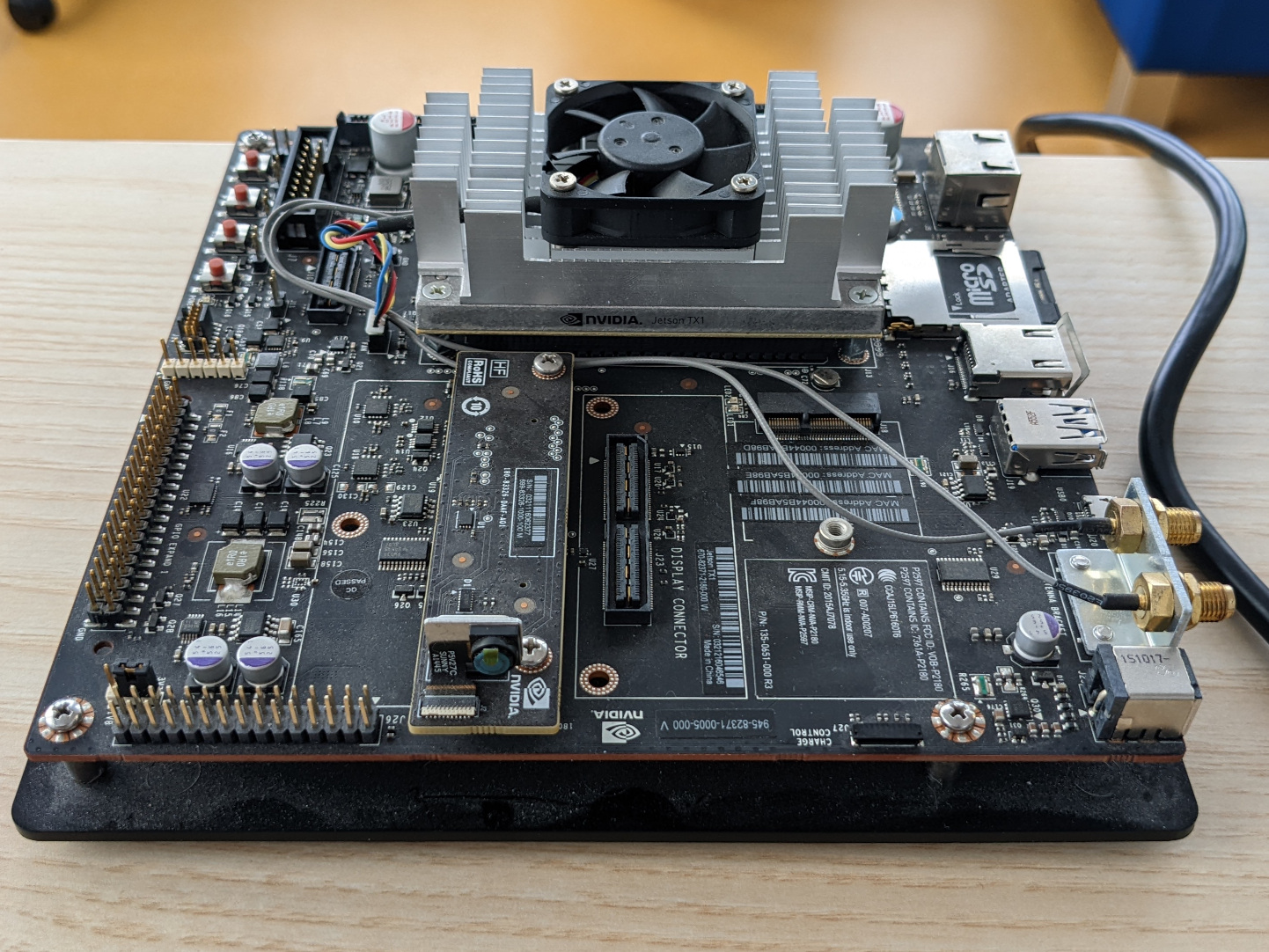} \\
         (a) & (b) & (c) 
    \end{tabular}
    
    \caption{Embedded platforms used for the evaluation: (a) I.MX8QuadMax Multisensory Enablement Kit by NXP (I.MX8~MEK), (b) Toradex Apalis I.MX8 board (I.MX8~Ixora), (c) Nvidia Jetson TX2 Developer Kit (TX2).}
    \label{fig:platforms}
\end{figure}

We selected modern high-performing \AbbrMPSOC{}s for the evaluation, namely I.MX~8QuadMax by NXP~\cite{2021:NXP-IMX8QM} and Nvidia Tegra X2 T186 \cite{2017:NVIDIA-TX2-TRM}. Both of these are based on ARM big.LITTLE heterogeneous architecture hosting two CPU clusters including so-called \emph{high-performing} and \emph{energy-efficient} cores. 

The I.MX~8QuadMax features four ARM Cortex-A53 cores and two ARM Cortex-A72 cores. Each of the cores has \SI{32}{\kilo\byte} data cache, and each cluster has \SI{1}{\mega\byte} L2 cache. We set the clock frequency of each cluster to the highest values, which is \SI{1200}{\mega\hertz} for the A53 cluster, and \SI{1600}{\mega\hertz} for the A72 cluster, respectively.

Similarly, Nvidia Tegra X2 T186 hosts four energy cores and two high-performing cores, which are of ARM Cortex-A57 architecture and Nvidia Denver architecture, respectively. Each A57 core has \SI{32}{\kilo\byte} data cache, and each Denver core has \SI{64}{\kilo\byte} data cache. The size of the L2 cache of each cluster is \SI{2}{\mega\byte}. 
We set the clock frequency of both clusters to \SI{2035}{\mega\hertz}.

In our testbed, we have two boards with I.MX8, namely I.MX8QuadMax Multisensory Enablement Kit (MEK) \cite{imx8-mek}, and Ixora carrier board with Toradex Apalis I.MX8 module \cite{imx8-toradex}. In the further text, these platforms are denoted as I.MX8~MEK and I.MX8~Ixora, respectively. Besides their different form factor and PCB layout, the first one has an Aluminum heat sink mounted on the chip while the latter has none, but we cool it by airflow from an external fan. In this way, the latter chip can be observed by a Workswell infrared camera~\cite{workswellWIC336}. 
We have extended both I.MX8 boards with external power meters.
Besides I.MX8, we have Nvidia \AbbrMPSOC{}, which is mounted on NVIDIA Jetson TX2 Developer Kit carrier board \cite{TX2devkit}. We henceforth denote this platform simply as TX2. 
A part of our testbed is shown in \Cref{fig:platforms}.
The configuration and used sensors are described in more detail in \cite{2020:Sojka}.

\subsection{Benchmarking Kernels} \label{sec:benchmarking-kernels}

To mimic the safety-critical workloads used in avionics and other similar domains, we use a set of relatively simple applications (kernels) written in C.
The set contains selected kernels based on EEMBC AutoBench~2.0~\cite{2022:EEMBC-Autobench} together with custom memory stressing tool \emph{membench} and OpenGL-like software rendering tool \emph{tinyrenderer} \cite{2020:Sojka}. 
%
We use tinyrenderer in two configurations -- rendering boggie objects \hbox{(\emph{-boggie})} and diablo objects \hbox{(\emph{-diablo})}.

AutoBench is a general-purpose benchmark set containing generic workload tests, as well as automotive and signal-processing algorithms. We use twelve of its kernels including: \emph{a2time} (angle to time conversion), \emph{aifirf} (finite impulse response filter), \emph{bitmnp} (bit manipulation), \emph{canrdr} (CAN remote data request), \emph{idctrn} (inverse discrete cosine transform), \emph{iirflt} (infinite impulse response filter), \emph{matrix} (matrix arithmetic), \emph{pntrch} (pointer chasing), \emph{puwmod} (pulse width modulation), \emph{rspeed} (road speed calculation), \emph{tblook} (table lookup and interpolation), and \emph{ttsprk} (tooth to spark). Each benchmark is used in two variants, i.e. \emph{-4K} and \emph{-4M}, representing two different input data sizes (\SI{4}{\kilo\byte} and \SI{4}{\mega\byte}). Further information about the benchmarks can be found in \cite{2009:Poovey}.

\emph{Membench} is a tool that stresses the memory hierarchy. It can be configured in many ways. We use it in three different configurations with respect to the working set size (\AbbrWSS), i.e. \emph{-1K}, \emph{-1M} and \emph{-4M}, representing \AbbrWSS{} of \SI{1}{\kilo\byte}, \SI{1}{\mega\byte}, and \SI{4}{\mega\byte}, respectively. Further, we test both sequential (\emph{-S}) and random (\emph{-R}) memory accesses in both read-only (\emph{-RO}) and read-and-write (\emph{-RW}) variants. Therefore, we have twelve membench kernels in our benchmark set. 

Each of the kernels ($12\times2$ autobench, $12$ membench, $2$ tinyrenderer) is wrapped inside of an infinite loop. A single iteration represents one execution of the kernel. We report the iterations per second (\AbbrIPS) of each kernel (executed on a single core, without any interference) for each tested hardware platform in \Cref{tab:kernel-ips} in \ref{app:ips}. 

\Cref{fig:ips-ratio} shows the relative speedup $s$ on a high-performing (big) cluster compared to the energy-efficient (little) cluster, i.e., the ratio between runtimes $e$ on these two clusters normalized by their frequencies $f$. We calculate the relative speedup $s$ as:

\begin{equation}
  s
  =\frac{\frac{e_{\text{little}}}{f_{\text{little}}}}{\frac{e_{\text{big}}}{f_{\text{big}}}}
  =\frac{\text{IPS}_{\text{big}}f_{\text{big}}}{\text{IPS}_{\text{little}}f_{\text{little}}}.
\end{equation}

We observe that the big cluster of I.MX8 (TX2) platform is, on average, about $2.8 \times$ and ($1.3 \times$) more performant than the little one.

\begin{figure}
    \centering


\newcommand{\Data}{data/speedup.csv}

\begin{tikzpicture}
 \tikzset{every mark/.append style={scale=2}}

    \begin{axis}[width=\textwidth,
        height=6cm,
        ymin=0,
        grid=major,
        xmin=-0.1,
        xmax=37.1,
        xtick=data,
        table/col sep=comma,
        xticklabels from table={\Data}{benchmark},
        xticklabel style={rotate=90},
        legend style={
            at={([yshift=5pt]0.5,1)},
            anchor=south,
            font=\scriptsize,
        },
        legend columns=3,
        xlabel=\footnotesize Benchmark,
        ylabel=\footnotesize Relative speedup $s$,
	  ytick={0,1,2,3,4},
	  ticklabel style = {font=\scriptsize}
]

    \addplot[Clr1,thick,mark=x] table [x expr=\coordindex,y={runtime_imx8a}, col sep=comma] {\Data}; \addlegendentry{I.MX8 MEK}

    \addplot[Clr2,thick,mark=square] table [x expr=\coordindex,y={runtime_imx8b}, col sep=comma] {\Data}; \addlegendentry{I.MX8 Ixora}

    \addplot[Clr3,thick,mark=triangle] table [x expr=\coordindex,y={runtime_tx2}, col sep=comma] {\Data}; \addlegendentry{TX2}

    \end{axis}

\end{tikzpicture}
    
    \caption{Relative speedup on a CPU from the high-performing cluster.}
    \label{fig:ips-ratio}
\end{figure}


\section{Experimental Evaluation and Results} \label{sec:experiments}

To evaluate the described power models and the optimization methods, we conduct a series of experiments on three physical platforms introduced in~\cref{sec:physical-hardware}.

First, we experimentally justify our selection of the use thermal model (\cref{sec:thermal-model-analysis}), transition from it to a simpler power model (\cref{sec:relat-betw-therm}) and experimentally identify platform and task characteristics (\cref{sec:platf-task-char}).
Then, we assess the quality of the proposed power models in \cref{sec:power-model-eval}. Finally, we compare the optimization methods with respect to the capability of reducing the peak temperature (\cref{sec:experiments-comparison}) and based on their computational complexity and scalability (\cref{sec:experiments-performance}).

\subsection{Thermal Experiments with Used Platforms} \label{sec:thermal-model-analysis}

To justify the selection a steady-state, single-output thermal model, we perform a set of experiments to analyze the platform's thermal dynamics (\cref{sec:thermal-dynamics}), and assess the relationship between the temperatures of little and big CPU clusters (\cref{sec:single-multiple-output}).

\subsubsection{Thermal Dynamics and the Major Frame Length}
\label{sec:thermal-dynamics}

The steady-state thermal model is sufficient if the temporal parameters of the workload are significantly shorter than time constants of the thermal dynamics~\cite{2011:Chantem}. Since typical lengths of the major fame used in avionics applications are less than one second, we need to determine whether thermal dynamics of our platforms is slower. We perform the following experiment:  a schedule containing two windows of the same length is created. These two windows constitute the major frame. In the first window, all the cores are loaded, executing some workload (here \emph{pntrch-4M}), whereas, in the second window, all cores are idling. We alternately execute these two windows and monitor the temperature and power consumption of the platform. We create three instances, which differ in the major frame length and \Cref{fig:period-vs-temp} shows them with different colors. The first (denoted with suffix \emph{-1s}) has the major frame length equal to \SI{1}{\second} (each window is \SI{500}{\milli\second} long), the second (\emph{-10s}) has the major frame length \SI{10}{\second}, and the third (\emph{-100s}) has the major frame length \SI{100}{\second}. The resulting temperatures measured in the proximity of the big cluster are shown in \Cref{fig:period-vs-temp}.

\begin{figure}
    \centering
      \tikzset{every mark/.append style={scale=0.5}}
  \pgfplotsset{style DataStyle/.style={}}
  
  \pgfplotsset{
    compat=1.15,
    DataStyle/.style={
        mark size=0.5pt,
        mark=*,
        join=round
    }
}

    \begin{tikzpicture}
    \begin{groupplot}[
    group style={
        columns=3,
        rows=2,
        group name=plots,
        x descriptions at=edge bottom,
        y descriptions at=edge left,
        horizontal sep=10pt,
        vertical sep=5pt,
    },
    ylabel={Temperature [\si{\celsius}]},
    ymin=42, ymax=65,
    xmin=0, xmax=300,
    xlabel={Time [\si{\second}]},
    enlarge x limits={abs=5},
    width=0.33\textwidth,
    height=5cm,
    grid=major,
    ticklabel style = {font=\scriptsize},
    legend columns=3,
    legend style={
            at={([yshift=45pt]0.5,1)},
            anchor=south,
            font=\scriptsize,
    },
    title style={align=center}]

    \nextgroupplot[title={I.MX8 MEK\\(heat sink, no air flow)}]
        \addplot[Clr1,DataStyle] table [x={time/s},y={CPU_1_temp/C}, col sep = comma] {data/imx8a/pntrch-4M-50_50_1-transformed.csv};
        \addplot[Clr2,DataStyle] table [x={time/s},y={CPU_1_temp/C},col sep = comma] {data/imx8a/pntrch-4M-50_50_10-transformed.csv};
        \addplot[Clr4,DataStyle] table [x={time/s},y={CPU_1_temp/C},col sep = comma] {data/imx8a/pntrch-4M-50_50_100-transformed.csv};
        
    \nextgroupplot[title={I.MX8 Ixora\\(no heat sink, air flow)}]
        \addplot[Clr1,DataStyle] table [x={time/s},y={CPU_1_temp/C}, col sep = comma] {data/imx8b/pntrch-4M-50_50_1-transformed.csv}; \addlegendentry{pntrch-4M-1s}
        \addplot[Clr2,DataStyle] table [x={time/s},y={CPU_1_temp/C},col sep = comma] {data/imx8b/pntrch-4M-50_50_10-transformed.csv}; \addlegendentry{pntrch-4M-10s}
        \addplot[Clr4,DataStyle] table [x={time/s},y={CPU_1_temp/C},col sep = comma] {data/imx8b/pntrch-4M-50_50_100-transformed.csv}; \addlegendentry{pntrch-4M-100s};
        
    \nextgroupplot[title={TX2\\(heat sink, no air flow)}]
        \addplot[Clr1,DataStyle] table [x={time/s},y={CPU_1_temp/C}, col sep = comma] {data/tx2/pntrch-4M-50_50_1-transformed.csv};
        \addplot[Clr2,DataStyle] table [x={time/s},y={CPU_1_temp/C},col sep = comma] {data/tx2/pntrch-4M-50_50_10-transformed.csv};
        \addplot[Clr4,DataStyle] table [x={time/s},y={CPU_1_temp/C},col sep = comma] {data/tx2/pntrch-4M-50_50_100-transformed.csv};
    
    \nextgroupplot[ylabel={Power [\si{\watt}]}, height=3cm, ymin=0, ymax=11]
        \addplot[Clr1,DataStyle,only marks] table [x ={time/s},y={power/W}, col sep = comma] {data/imx8a/pntrch-4M-50_50_1-transformed.csv};
        \addplot[Clr2,DataStyle,only marks] table [x ={time/s},y={power/W}, col sep = comma] {data/imx8a/pntrch-4M-50_50_10-transformed.csv};
        \addplot[Clr4,DataStyle,only marks] table [x ={time/s},y={power/W}, col sep = comma] {data/imx8a/pntrch-4M-50_50_100-transformed.csv};

    \nextgroupplot[height=3cm, ymin=0, ymax=11]
        \addplot[Clr1,DataStyle,only marks] table [x ={time/s},y={power/W}, col sep = comma] {data/imx8b/pntrch-4M-50_50_1-transformed.csv};
        \addplot[Clr2,DataStyle,only marks] table [x ={time/s},y={power/W}, col sep = comma] {data/imx8b/pntrch-4M-50_50_10-transformed.csv};
        \addplot[Clr4,DataStyle,only marks] table [x ={time/s},y={power/W}, col sep = comma] {data/imx8b/pntrch-4M-50_50_100-transformed.csv};

    \nextgroupplot[height=3cm, ymin=0, ymax=11]
        \addplot[Clr1,DataStyle,only marks] table [x ={time/s},y={power/W}, col sep = comma] {data/tx2/pntrch-4M-50_50_1-transformed.csv};
        \addplot[Clr2,DataStyle,only marks] table [x ={time/s},y={power/W}, col sep = comma] {data/tx2/pntrch-4M-50_50_10-transformed.csv};
        \addplot[Clr4,DataStyle,only marks] table [x ={time/s},y={power/W}, col sep = comma] {data/tx2/pntrch-4M-50_50_100-transformed.csv};
    
     \end{groupplot}
 
  \end{tikzpicture}
    
    \caption{Influence of the on-chip temperature near the high-performing cluster on the major frame length for three instances alternating between computing and idling.}
    \label{fig:period-vs-temp}
\end{figure}

Indeed, when the major frame length of the instance is long enough, such as in the \emph{-100s} case, we clearly observe the heating and cooling curves corresponding to the individual windows.
However, for our use-case (\emph{-1s}), we see that the temperature is almost constant.

Note that the power consumption of both platforms based on I.MX8 is nearly the same; however, their thermal trajectories differ significantly due to different physical parameters (heat sink versus no heat sink, with/without airflow).

\subsubsection{Spatial Distribution of On-Chip Temperatures}
\label{sec:single-multiple-output}

When executing a workload, different parts of the chip start to produce heat. Ideally, we would like to monitor the temperature of each core. However, per-core temperature monitoring might not be possible for many platforms, including ours. Our three platforms provide us with just several temperature sensors associated with the major thermal zones (little cluster, big cluster, PMIC, GPU, etc.). We visualize the temperatures measured for the \emph{pntrch-4M-100s} benchmark (the one used in the previous section) near little and big clusters in \Cref{fig:temp-of-both-clusters}.

\begin{figure}
    \centering
        \tikzset{every mark/.append style={scale=0.5}}
    \begin{tikzpicture}
    \begin{groupplot}[
    group style={
        columns=3,
        rows=1,
        group name=plots,
        x descriptions at=edge bottom,
        y descriptions at=edge left,
        horizontal sep=10pt,
        vertical sep=5pt,
    },
    ticklabel style = {font=\scriptsize},
    ylabel={Temperature [\si{\celsius}]},
    ymin=42, ymax=65,
    xmin=0, xmax=300,
    xlabel={Time [\si{\second}]},
    enlarge x limits={abs=5},
    width=0.33\textwidth,
    height=5cm,
    grid=major,
    legend columns=3,
    legend style={
            at={([yshift=25pt]0.5,1)},
            anchor=south,
            font=\scriptsize,
    }]

    \nextgroupplot[title={I.MX8 MEK}]
        \addplot[Clr1,thick,mark=x] table [x={time/s},y={CPU_0_temp/C},col sep = comma] {data/imx8a/pntrch-4M-50_50_100-transformed.csv};
        \addplot[Clr4,thick,mark=x] table [x={time/s},y={CPU_1_temp/C},col sep = comma] {data/imx8a/pntrch-4M-50_50_100-transformed.csv};

    \nextgroupplot[title={I.MX8 Ixora}]
        \addplot[Clr1,thick,mark=x] table [x={time/s},y={CPU_0_temp/C},col sep = comma] {data/imx8b/pntrch-4M-50_50_100-transformed.csv}; \addlegendentry{Little cluster}
        \addplot[Clr4,thick,mark=x] table [x={time/s},y={CPU_1_temp/C},col sep = comma] {data/imx8b/pntrch-4M-50_50_100-transformed.csv}; \addlegendentry{Big cluster}
        
    \nextgroupplot[title=TX2]
        \addplot[Clr1,thick,mark=x] table [x={time/s},y={CPU_0_temp/C},col sep = comma] {data/tx2/pntrch-4M-50_50_100-transformed.csv};
        \addplot[Clr4,thick,mark=x] table [x={time/s},y={CPU_1_temp/C},col sep = comma] {data/tx2/pntrch-4M-50_50_100-transformed.csv};
     \end{groupplot}
 
  \end{tikzpicture}
    
    \caption{Temperatures obtained for \emph{pntrch-4M-100s} from on-chip sensors near little and big cluster thermal zones.}
    \label{fig:temp-of-both-clusters}
\end{figure}

We observe that the temperature difference on I.MX8~Ixora is smaller compared to I.MX8~MEK, because of the absence of a heat sink on the Ixora board and active cooling that is employed.
Considering the TX2 platform, we observe that both thermal zones report the same value.
This might be caused by a massive heatsink, combined with the imprecision of the sensors and their possible spatial proximity.

To further investigate the thermal behavior near the CPU clusters, we look at I.MX8~Ixora using the Thermal camera. We execute \emph{pntrch-4M} on all cores of each cluster and compare the resulting images. \Cref{fig:thermocam} shows the spatial on-chip temperature $T(x,y)$, where the $x$ and $y$ coordinates are in pixels (each pixel corresponds to \SI{0.29}{\milli\meter}). Also, we show the heat sources on a chip $h(x,y)$, where $h(x,y) = \max\{0, -\kappa \nabla^2 T(x,y) \}$ is a positive part of negative Laplacian of $T(x,y)$ scaled by factor $\kappa > 0$, which follows from heat diffusion equation as explained in \cite{2020:zhang,2020:Sojka}.

\begin{figure}
    \centering
    \includegraphics[width=\textwidth]{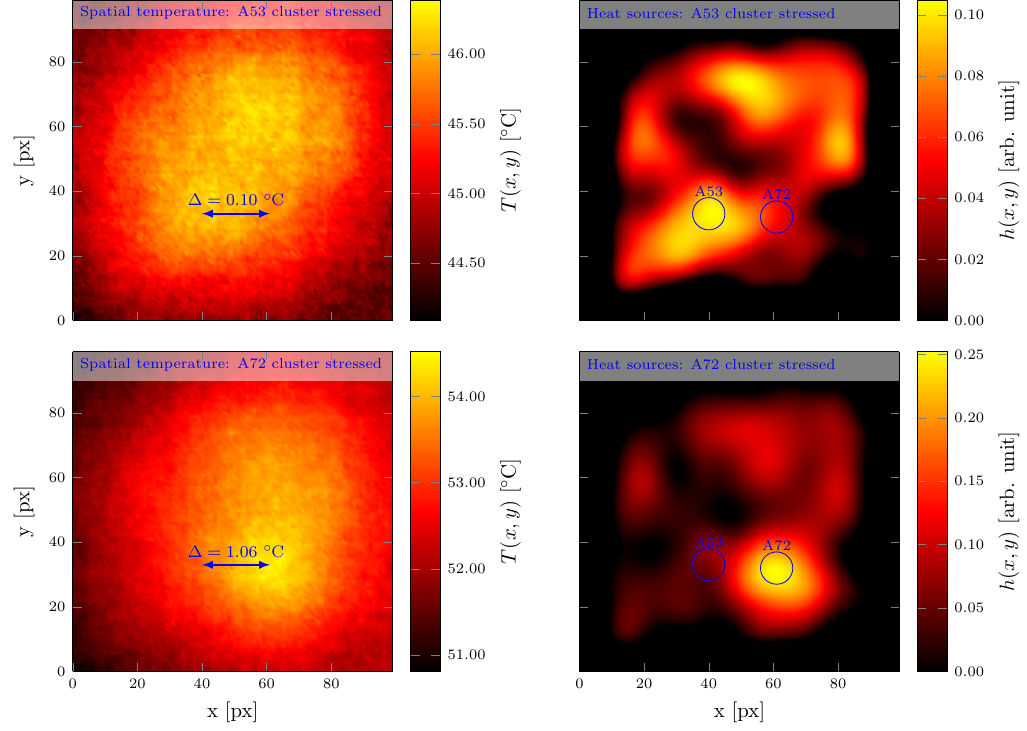}
    \caption{Spatial on-chip temperature $T(x,y)$ on the left and hot spots $h(x,y)$ on the right of I.MX8~Ixora with little (A53) cluster stressed at the top, and big (A72) cluster stressed at the bottom.}
    \label{fig:thermocam}
\end{figure}

\Cref{fig:thermocam} shows that the big cluster is heating the platform much more (the peak of $h(x,y)$ is about $2.5 \times$ higher) compared to the little one. Also, the left part of the figure shows how the on-chip heat spreader distributes the heat from the heat source to the borders of the chip. When only the little cluster is executing the workload, the difference between the individual cluster zones' temperatures is nearly negligible. When the big cluster is performing the computations, the difference is more apparent, but still only about \SI{1}{\celsius} for this particular workload.

To summarize the observations: although we see some differences between the temperatures  measured in the vicinity of the individual clusters, both of their thermal trajectories are similar, as seen in \Cref{fig:temp-of-both-clusters}. Due to the heat spreader and relative proximity of both clusters, the change of the temperature near one of the clusters influences the temperature near the other one as shown in \Cref{fig:thermocam}. Taking that into account, we decided to model only the temperature near the big cluster, which is thermally dominant.

\subsection{Relation Between Steady-State Temperature and Average Power Consumption}
\label{sec:relat-betw-therm}

In \cref{sec:thermal-to-power}, we decided to use average power consumption instead of steady-state temperature. To justify this decision, we performed experiments comparing these quantities for various benchmarks (both memory and CPU-bound) executed on our platforms. The results in \Cref{fig:power_vs_temp} show that both measured quantities are strongly correlated, approximating the linear relation derived in \cref{eq:tp-linear}.

\begin{figure}
    \centering
    \pgfplotsset{style DataStyle/.style={}}
\pgfplotsset{
    compat=1.15,
    DataStyle/.style={
        mark size=0.5pt,
        mark=*,
        join=round
    }
} 
\pgfplotscreateplotcyclelist{ClrList}{
{Clr1},
{Clr2},
{Clr3},
}

    



\begin{tikzpicture}
    \begin{axis}[
    ylabel={Steady-State Temperature [\si{\celsius}]},
    xlabel={Average Power [\si{\watt}]},
    cycle list name=ClrList,
    xmin=2,
    xmax=12,
    width=0.6\textwidth,
    height=6cm,
    grid=major,
    legend pos=south east,
    legend cell align={left},
    legend style={font=\scriptsize, },
    ticklabel style = {font=\scriptsize},
    legend image post style={scale=3},
    ]

    \addplot+[DataStyle,only marks] table [y ={T},x={P}, col sep = comma] {data/temp-and-power/imxa.csv};
    \addlegendentry{I.MX8~MEK}
    
    \addplot+[DataStyle,only marks] table [y ={T},x={P}, col sep = comma] {data/temp-and-power/imxb.csv};
    \addlegendentry{I.MX8~Ixora}

    \addplot+[DataStyle,only marks] table [y ={T},x={P}, col sep = comma] {data/temp-and-power/tx.csv};
    \addlegendentry{TX2}

    \end{axis}
\end{tikzpicture}
    \caption{Average power and steady-state temperature (measured, at the thermal zone near the big cluster) of various benchmarks executed on tested platforms.}
    \label{fig:power_vs_temp}
\end{figure}

\subsection{Platform and Task Characteristics}
\label{sec:platf-task-char}

To evaluate the effectiveness of the proposed methods, we need to determine the platform and task characteristics used by the methods. The simplest characteristic to obtain is the platform idle power consumption $\PIdle$, which was measured with the connected power meters. The results are listed in \Cref{tab:p-idle}.

\begin{table}[ht]
    \centering
    \caption{Idle power consumption $\PIdle$ of tested platforms.}
    \begin{scriptsize}    
    \begin{tabular}{cc}
        \toprule
        Platform & $\PIdle$ [\si{\watt}]  \\
        \midrule
        I.MX8~MEK & 5.5 \\
        I.MX8~Ixora & 5.5 \\
        TX2 & 2.6 \\
        \bottomrule
    \end{tabular}
    \end{scriptsize}
    
    \label{tab:p-idle}
\end{table}

The task characteristics coefficients $\TaskCoefOffset{\TaskIdx}{\ResIdx}$ and $\TaskCoefSlope{\TaskIdx}{\ResIdx}$ for the \ModelSumMax{} model were obtained following our methodology introduced in~\cite{2021:Benedikt}. With it, we identify the coefficients for all benchmarks on each tested platform. Their values are visualized in \Cref{fig:task-coefficients} (and further listed in \ref{app:kernel-chars-values}). Note that the sum of $\TaskCoefOffset{\TaskIdx}{\ResIdx}$ and $\TaskCoefSlope{\TaskIdx}{\ResIdx}$ (i.e., the height of the bar in \Cref{fig:task-coefficients}) represents the increase the power consumption of the platform w.r.t. $\PIdle$ when executing the benchmark on a single core of cluster $\Res{\ResIdx}$.

\begin{figure}
    \centering
    \newcommand{\DataLittle}{data/parameters_little_nneg.csv}
\newcommand{\DataBig}{data/parameters_big_nneg.csv}

\begin{tikzpicture}
 \begin{groupplot}[
    group style={
        columns=1,
        rows=3,
        x descriptions at=edge bottom,
        y descriptions at=edge left,
        horizontal sep=10pt,
        vertical sep=5pt,
    },
  ybar stacked,
  stack negative=separate,
  width=\textwidth,
  height=6cm,
  legend style={
            at={([yshift=5pt]0.5,1)},
            anchor=south,
            font=\scriptsize,
        },
  legend columns= 2,
  ylabel={\footnotesize Coefficient value [\si{\watt}]},
  xtick=data,  
  xmin=-3.5,
  xmax=37.5,
  ymin=-4.2,
  ymax=4.2,
  table/col sep=comma,
  xticklabel style={rotate=90,font=\scriptsize},
  xticklabels from table={\DataLittle}{benchmark},
  height=6cm,
  grid=major,
  ytick={-4,-2,0,2,4},
  yticklabels={4,2,0,2,4},
ticklabel style = {font=\scriptsize}
]

\pgfplotsinvokeforeach{imxa,imxb,tx}{
    \ifthenelse{\equal{#1}{imxa}}
    {\nextgroupplot[legend entries={$\TaskCoefOffset{\TaskIdx}{\ResIdx}$,$\TaskCoefSlope{\TaskIdx}{\ResIdx}$}]}
    {\nextgroupplot[]}

    \ifthenelse{\equal{#1}{imxa}}{\node[anchor=north west,fill=white, fill opacity=0.7, text opacity=1] at (axis cs: 0,3.8) {\footnotesize I.MX8~MEK};}{}
    \ifthenelse{\equal{#1}{imxb}}{\node[anchor=north west,fill=white, fill opacity=0.7, text opacity=1] at (axis cs: 0,3.8) {\footnotesize I.MX8~Ixora};}{}
    \ifthenelse{\equal{#1}{tx}}{\node[anchor=north west,fill=white, fill opacity=0.7, text opacity=1] at (axis cs: 0,3.8) {\footnotesize TX2};}{}

     \addplot[Clr1,fill=Clr1, bar width=5pt] table [x expr=\coordindex,y={intercept#1}, col sep=comma] {\DataLittle};  
    
     \addplot[Clr2,fill=Clr2, bar width=5pt] table [x expr=\coordindex,y={slope#1}, col sep=comma] {\DataLittle};  
    
    \addplot[Clr1,fill=Clr1, bar width=5pt, opacity=0.7] table [x expr=\coordindex,y expr=-\thisrow{intercept#1}, col sep=comma] {\DataBig};  
    
     \addplot[Clr2,fill=Clr2, bar width=5pt, opacity=0.7] table [x expr=\coordindex,y expr=-\thisrow{slope#1}, col sep=comma] {\DataBig};  
    
    \draw[thick] (axis cs: -4,0) -- (axis cs: 40, 0);
    \draw[-latex, thick] (axis cs: -1.5,0) -- node[sloped,anchor=south,fill=white,fill opacity=0.7, text opacity=1] {\scriptsize little cluster} (axis cs: -1.5, 4.2);
    \draw[latex-, thick] (axis cs: -1.5,-4.2) -- node[sloped,anchor=south,fill=white,fill opacity=0.7, text opacity=1] {\scriptsize big cluster} (axis cs: -1.5, 0);
}

\end{groupplot}
\end{tikzpicture}
    
    \caption{Values of task characteristics coefficients $\TaskCoefOffset{\TaskIdx}{\ResIdx}$ and $\TaskCoefSlope{\TaskIdx}{\ResIdx}$ of tested benchmarks on little (top semi-axis) and big (bottom semi-axis) clusters.}
    \label{fig:task-coefficients}
\end{figure}

\subsubsection{Identification of Regression Coefficients}
\label{sec:ident-regr-coeff}

 To identify the regression coefficients for the \ModelLR{} model, we created 1000 unique instances (each representing one interval $I$), which were randomly populated with the benchmarking kernels described in \Cref{sec:benchmarking-kernels}. Specifically, each interval was \SI{1}{\second} long and contained from zero (all cores idling) up to 6 (all cores processing) kernels randomly picked from the set.  All these instances were executed on all tested platforms (each for \SI{180}{\second}), and the average power consumption was measured. The identified coefficients (i.e., elements of vectors $ \bm{\beta}_{\ResIdx}$ for each cluster $\Res{\ResIdx} \in \ResSet$) were obtained from the measured data by linear regression and are reported in \Cref{tab:regression-coefficients}.

\begin{table}[ht]
    \centering
    \caption{Regression coefficients identified for all tested platforms and coefficient of determination $R^2$.}
    \label{tab:regression-coefficients}
    \begin{scriptsize}
    \begin{tabular}{rccccc}
        \toprule
        & \multicolumn{2}{c}{little cluster ($\bm{\beta_1}$)} & \multicolumn{2}{c}{big cluster ($\bm{\beta_2}$)} \\
         \cmidrule(r){2-3} \cmidrule(l){4-5}
        platform & $\beta_{1,1}$ & $\beta_{2,1}$ & $\beta_{1,2}$ & $\beta_{2,2}$ & $R^2$ \\
        \midrule
       I.MX8 MEK    & 1.205 & 0.270  & 0.969 & 0.456 & 0.822 \\
       I.MX8 Ixora  & 1.227 & 0.232  & 0.981 & 0.420 & 0.814 \\
       TX2          & 0.857 & 0.648  & 0.946 & 0.801 & 0.974 \\
       \midrule
       \multicolumn{1}{r}{\scriptsize corresponding independent var. $\rightarrow$} & $\TaskCoefSlope{\TaskIdx}{1}$ & $\TaskCoefOffset{\TaskIdx}{1}$ & $\TaskCoefSlope{\TaskIdx}{2}$ & $\TaskCoefOffset{\TaskIdx}{2}$ \\
       \bottomrule
    \end{tabular}
    \end{scriptsize}
\end{table}

\subsection{Power Model Evaluation} \label{sec:power-model-eval}

Now, we can evaluate the accuracy of the proposed \ModelSumMax{}, \ModelLR{} and \texttt{LR-UB} power models.
For each tested platform, we generate one thousand instances of \emph{CPU-bound workload} (all kernels except membench) and one thousand instances of \emph{mixed workload} (all kernels including membench), that is two thousand distinct instances in total.
The workload consists of a single periodically repeated window with a length of \SI{1}{\second}. In that window, each CPU executes either nothing (probability 0.5) or a random kernel. The kernels are executed for duration uniformly selected from interval \SI{1}{\milli\second} to \SI{1000}{\milli\second}.

Each instance was executed for \SI{60}{\second}; this gives more than four days (100 hours)  of measured data in total. The power consumption was sampled every \SI{10}{\milli\second}, and the average value was reported. Further, we calculated the predicted power by \ModelSumMax{}, \ModelLR{}, and \texttt{LR-UB} power models.

The results for mixed workload instances are shown in \cref{fig:power_evaluation}. The instances are sorted by the measured power consumption on I.MX8 MEK.
\cref{tab:power_mae} shows the mean absolute error of all power models on both types of workload. We observe that the lowest prediction error is achieved by the linear regression (\texttt{LR}) model. In relative terms (w.r.t. the idle power consumption), its error is \SI{4.3}{\percent}, \SI{4.5}{\percent} and \SI{13.3}{\percent} for I.MX8 MEX, I.MX8 Ixora and TX2, respectively. The \ModelSumMax{} model performed slightly worse, with average relative error of \SI{11.2}{\percent}, \SI{5.6}{\percent}, and \SI{16.0}{\percent} for I.MX8 MEX, I.MX8 Ixora and TX2. Finally, the \texttt{LR-UB} model failed to deliver satisfactory predictions; its relative error is \SI{24.3}{\percent}, \SI{19.3}{\percent}, and \SI{74.4}{\percent} for I.MX8 MEX, I.MX8 Ixora and TX2.

The trends are also clearly visible in \cref{fig:power_evaluation};  \ModelSumMax{} model is more pessimistic than the \ModelLR{} model, which is expected due to the max term in \cref{eq:max-sum-model}. However, it steadily provides an upper bound on the measured power consumption. Even though the \texttt{LR-UB} mostly provides an upper bound as well, it is not as tight.

\begin{figure}
    \centering
    \includegraphics[width=\textwidth]{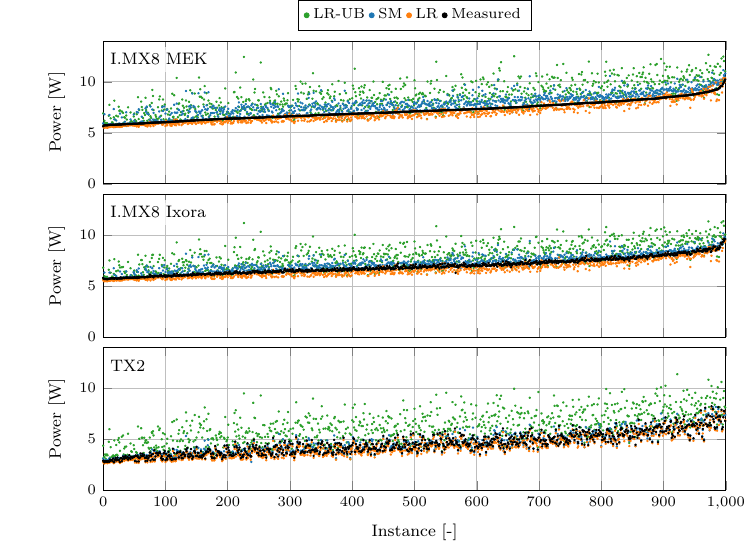}
    \caption{Measured and predicted power consumption of 1000 testing instances (mixed workload windows); instances are sorted by I.MX8~MEK measured power consumption.}
    \label{fig:power_evaluation}
\end{figure}

\begin{table}
    \centering
    \caption{Mean absolute error (in Watts) of the tested power models.} \label{tab:power_mae}
    \begin{scriptsize}
    \begin{tabular}{rrrrrrr}
    \toprule
         & \multicolumn{2}{c}{SM} & \multicolumn{2}{c}{LR} & \multicolumn{2}{c}{LR-UB} \\
         \cmidrule(r){2-3} \cmidrule(lr){4-5} \cmidrule(l){6-7}
        Platform & mixed & CPU & mixed & CPU & mixed & CPU \\
        \midrule
        I.MX8~MEK & 0.67 & 0.56 & \textbf{0.26} & \textbf{0.21} & 1.30 & 1.38 \\ 
        I.MX8~Ixora & 0.35 & 0.27 & \textbf{0.28} & \textbf{0.21} & 1.00 & 1.12 \\ 
        TX2 & \textbf{0.17} & 0.66 & 0.24 & \textbf{0.45} & 1.54 & 2.33 \\ 
        \bottomrule
    \end{tabular}
    \end{scriptsize}
\end{table}

\subsection{Optimization Methods Comparison} \label{sec:experiments-comparison}

Here, we discuss how well the power models integrate with the optimization. We compare the optimization methods on two types of workloads as in the previous section: CPU-bound and mixed. For each workload type, we construct six different instances. 
We generate 20 tasks; each of them executes a randomly selected kernel. Each task is assigned a randomly generated execution time on the big cluster in the range \SI{40}{\milli\second} to \SI{160}{\milli\second}. Execution time for the little cluster is scaled appropriately to perform the same work. The scaling coefficient is calculated from \cref{tab:kernel-ips}. The major frame length $\MF$ is calculated as $\MF = \frac{\TaskNum \cdot \bar{e}}{\kappa}$, where $\bar{e}$ is the average execution time across all clusters, $\TaskNum$ is the number of tasks (here 20), and $\kappa$ is empirical constant changing the tightness of the schedules (here set to $3.5$).

\begin{table}
    \centering
    \caption{List of compared optimization methods and corresponding power models.}
    \begin{tabular}{|c|c|c|}
    \hline
    Acronym & Power model & Optimization method \\ \hline
    \ModelILPSMOrig{} & \ModelSumMax{} &  ILP \\
    \ModelQPLR{} & \texttt{LR-UB} & QP \\
    \texttt{BB-LR} & \ModelLR{} & BB (generic alg.) \\
    \ModelBlackBoxSM{} & \ModelSumMax{} & BB (generic alg.) \\
    \ModelHeur{} & expected energy & greedy \\
    \ModelILPIdleMin{} & --- & ILP \\
    \ModelILPIdleMax{} & --- & ILP \\
    \hline
\end{tabular}
    \label{tab:opt-methods}
\end{table}

For each instance, all optimization methods, as listed in \cref{tab:opt-methods}, were executed to generate a schedule for each platform.
For better comparison, we execute the black-box optimizer with both \ModelSumMax{} and \ModelLR{} power models.
The solving time limit was set to \SI{300}{\second} per instance.
The schedules found for the first instance are illustrated in \ref{app:example-schedules}.

During the experiment, each schedule was executed on the respective platform for \SI{30}{\minute}; this gives 42 hours of measured data per platform. We measured the average power consumption and steady-state temperature. The power offset ($P_{\text{measured}} - \PIdle$) is reported in \cref{tab:experiment-power-offset} (the rows are sorted by the average power consumption on I.MX8~MEK). The \ModelILPSMOrig{} method achieved the best results in almost all cases. The difference from the lowest result is negligible in the few cases where it was not the best. Slightly worse, but still good results were obtained by the \ModelBlackBoxSM{} method. One practical difference between these two methods is that \ModelILPSMOrig{} requires an ILP solver (here commercial Gurobi solver~\cite{gurobi}) for its operation, while \ModelBlackBoxSM{} can be implemented with freely available tools.

An interesting observation is that the best results are obtained with the \ModelSumMax{} power model. Recall that the most accurate power model was \ModelLR{}, not \ModelSumMax{}. We account that to the fact that even though the \ModelSumMax{} is systematically overestimating the power consumption, it is consistent in a sense that windows with higher predicted power consumption indeed consume more than windows with lower predicted power consumption.

\begin{table}[htb]
    \centering
    \begin{scriptsize}
    \sisetup{round-mode=places, round-precision=2}
    
    \caption{Power offset ($P_{\text{measured}} - \PIdle$ [\si{\watt}]) observed for six instances with mixed workloads and six instances with CPU-bound workloads on tested platforms.} \label{tab:experiment-power-offset}
    
    \begin{tabular}{r|cccccc|cccccc|c}
    \toprule
     & \multicolumn{6}{c}{Mixed workloads instances} & \multicolumn{6}{c}{CPU-bound instances}  \\ \cmidrule(lr){2-7} \cmidrule(lr){8-13}
     Method & 1 & 2 & 3 & 4 & 5 & 6  & 1 & 2 & 3 & 4 & 5 & 6 & average \\
    \midrule
     & \multicolumn{12}{c}{I.MX8~MEK} \\
      \begin{filecontents*}{data/power_imx8a_relative.csv}
method,inst0-all,inst1-all,inst2-all,inst3-all,inst4-all,inst5-all,inst0-cpu,inst1-cpu,inst2-cpu,inst3-cpu,inst4-cpu,inst5-cpu,average
ILP-SM,2.77,2.26,2.99,2.79,2.38,2.17,2.47,2.69,2.72,2.8,2.9,2.78,2.64
BB-SM,2.72,2.48,3.12,3.04,2.61,2.39,2.81,2.97,2.97,3.02,3.06,3.1,2.86
HEUR,3.16,2.71,3.17,3.2,3.12,2.44,2.9,3.01,3,3.26,3.16,3.36,3.04
BB-LR,3.09,2.82,3.36,3.22,3.01,2.63,2.94,3.03,2.98,3.18,2.93,3.31,3.04
ILP-IDLE-MIN,3.44,2.7,3.36,3.24,3.19,2.67,2.98,3.09,2.95,3.37,3.22,3.23,3.12
QP-LR-UB,3.61,2.84,3.5,3.48,3.24,2.44,3.12,3.51,3.11,3.35,3.3,3.21,3.23
ILP-IDLE-MAX,3.74,3.25,3.9,4.06,3.41,3.05,3.47,3.88,3.69,3.96,3.92,4.12,3.70
\end{filecontents*}
\csvreader[head to column names, late after line=\\]{data/power_imx8a_relative.csv}{}{\method & \num{\csvcolii} & \num{\csvcoliii} & \num{\csvcoliv} & \num{\csvcolv} & \num{\csvcolvi} & \num{\csvcolvii} & \num{\csvcolviii} & \num{\csvcolix} & \num{\csvcolx} & \num{\csvcolxi} & \num{\csvcolxii} & \num{\csvcolxiii} & \num{\csvcolxiv}}
      \midrule
       & \multicolumn{12}{c}{I.MX8~Ixora}  \\
      \begin{filecontents*}{data/power_imx8b_relative.csv}
method,inst0-all,inst1-all,inst2-all,inst3-all,inst4-all,inst5-all,inst0-cpu,inst1-cpu,inst2-cpu,inst3-cpu,inst4-cpu,inst5-cpu,average
ILP-SM,2.19,1.97,2.52,2.36,2,1.94,2.15,2.25,2.28,2.31,2.38,2.35,2.23
BB-SM,2.33,2.18,2.66,2.54,2.09,2.13,2.32,2.41,2.41,2.44,2.53,2.63,2.39
HEUR,2.54,2.21,2.58,2.53,2.46,2.23,2.45,2.41,2.42,2.59,2.56,2.77,2.48
BB-LR,2.54,2.28,2.65,2.63,2.44,2.12,2.4,2.48,2.41,2.51,2.51,2.75,2.48
ILP-IDLE-MIN,2.81,2.04,2.66,2.67,2.59,2.23,2.38,2.51,2.42,2.72,2.56,2.61,2.52
QP-LR-UB,2.69,2.24,2.87,2.86,2.37,2.05,2.51,2.9,2.69,2.69,2.71,2.94,2.63
ILP-IDLE-MAX,3.3,2.82,3.28,3.36,3.02,2.91,3.04,3.33,3.18,3.45,3.22,3.44,3.20
\end{filecontents*}
\csvreader[head to column names,late after line=\\]{data/power_imx8b_relative.csv}{}{\method & \num{\csvcolii} & \num{\csvcoliii} & \num{\csvcoliv} & \num{\csvcolv} & \num{\csvcolvi} & \num{\csvcolvii} & \num{\csvcolviii} & \num{\csvcolix} & \num{\csvcolx} & \num{\csvcolxi} & \num{\csvcolxii} & \num{\csvcolxiii} & \num{\csvcolxiv}}
      \midrule
      & \multicolumn{12}{c}{TX2}  \\
      \begin{filecontents*}{data/power_tx2_relative.csv}
method,inst0-all,inst1-all,inst2-all,inst3-all,inst4-all,inst5-all,inst0-cpu,inst1-cpu,inst2-cpu,inst3-cpu,inst4-cpu,inst5-cpu,average
ILP-SM,3.53,3.22,3.23,3.78,3.44,3.03,3.55,3.9,3.45,3.96,3.57,4.16,3.57
BB-SM,3.71,3.41,3.36,3.88,3.74,3.29,3.76,3.98,3.64,4.18,3.76,4.32,3.75
HEUR,3.59,3.27,3.25,4,3.7,3.11,3.74,3.91,3.54,4.15,3.73,4.19,3.68
BB-LR,3.7,3.3,3.43,3.92,3.71,3.26,3.78,3.98,3.71,4.1,3.88,4.26,3.75
ILP-IDLE-MIN,4.34,3.51,4.21,4.28,4.02,3.88,3.8,4.06,4.44,4.31,3.93,4.53,4.11
QP-LR-UB,3.65,3.35,3.28,3.82,3.61,3.13,3.63,3.98,3.55,3.97,3.57,4.29,3.65
ILP-IDLE-MAX,4.63,4.02,4.37,4.82,4.49,3.88,4.71,4.87,4.66,5.19,4.99,5.31,4.66
\end{filecontents*}
\csvreader[head to column names,late after line=\\]{data/power_tx2_relative.csv}{}{\method & \num{\csvcolii} & \num{\csvcoliii} & \num{\csvcoliv} & \num{\csvcolv} & \num{\csvcolvi} & \num{\csvcolvii} & \num{\csvcolviii} & \num{\csvcolix} & \num{\csvcolx} & \num{\csvcolxi} & \num{\csvcolxii} & \num{\csvcolxiii} & \num{\csvcolxiv}}
    \bottomrule
    \end{tabular}
    \end{scriptsize}
\end{table}

\Cref{fig:methods-temperatures} shows the temperatures near the individual clusters averaged over all six instances. 
We observe that the difference between the worst- and the best-performing methods are \SI{5.5}{\celsius}, \SI{4.9}{\celsius}, and \SI{3.6}{\celsius}, corresponding to \SI{22}{\percent}, \SI{19.6}{\percent}, and \SI{14.4}{\percent} differences relative to the ambient temperature (\SI{25}{\celsius}) for I.MX8~MEK, I.MX8~Ixora, and TX2. 

When comparing the greedy local heuristic \texttt{HEUR} with the \ModelILPSMOrig{} method, the exhaustive \ModelILPSMOrig{} can save, in average, up to \SI{1.6}{\celsius}, \SI{1.3}{\celsius}, and \SI{0.6}{\celsius} (corresponding to \SI{4.7}{\percent}, \SI{4.6}{\percent}, and \SI{1.8}{\percent}) for I.MX8~MEK, I.MX8~Ixora, and TX2, respectively.

\begin{figure}
    \centering
    \includegraphics[]{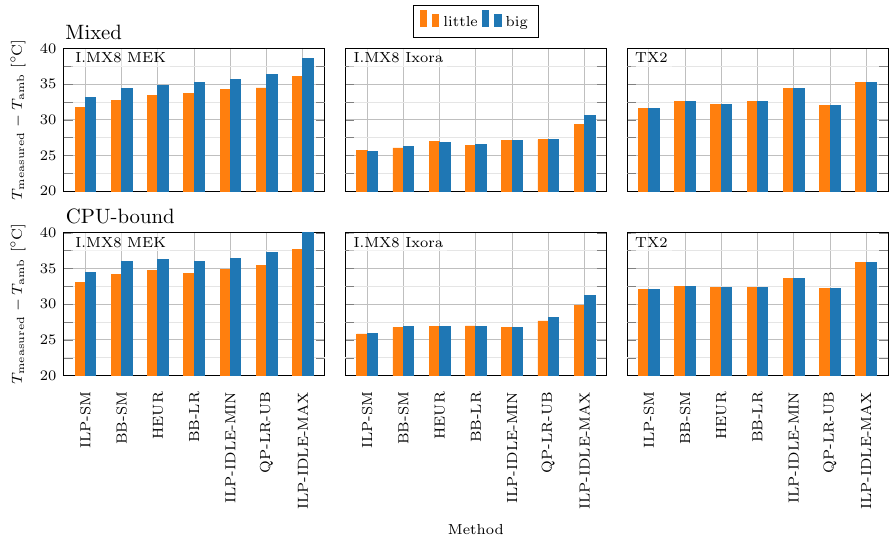}
    \caption{Average difference between the measured steady-state temperature $T_{\text{measured}}$ and the average ambient temperature $T_{\text{amb}}$ for tested optimization methods. Recall that Ixora's lower temperatures are caused by applied air flow.}
    \label{fig:methods-temperatures}
\end{figure}

\subsection{Performance evaluation} \label{sec:experiments-performance}

Finally, we evaluate the scalability of tested optimization methods. 
We study how the computation time increases with the increasing instance size corresponding to the number of tasks $\TaskNum$.

Ten instances are randomly generated for each $\TaskNum \in \{5, 10, \dots, 60\}$ (120 instances in total). Each optimization method is then executed for every instance; as some of the methods might be rather time-demanding for larger instance sizes, we limit the maximum computation time per instance to \SI{300}{\second}. We use the same generator as in \cref{sec:experiments-comparison}, but we perform the experiment only with the characteristics based on I.MX8~MEK. The outcome would be quite similar for the other platforms.

The average computation times for different values of $\TaskNum$ are shown in \cref{fig:scalability}. As the black-box optimizer (\ModelBlackBox{}) is programmed to randomly restart each time it converges to some solution, it always consumes all the provided time. Besides, the models globally optimizing the schedule w.r.t. the provided objective, i.e., \ModelQPLR{} and \ModelILPSMOrig{}, are the first to run out of time. Out of these two, the more complex model based on the quadratic programming (\ModelQPLR{}) is about $6 \times$ slower than \ModelILPSMOrig{} on instances with 15 and 20 tasks.
Comparing the global methods to the local one (\ModelHeur{}) on instances with 20 tasks, we see that the global methods \ModelILPSMOrig{} and \ModelQPLR{} need about $18 \times$ and $95 \times$ more time, respectively.
Performance of the heuristic method (\ModelHeur{}) is comparable with \ModelILPIdleMin{} on instances with $30$ and more tasks; for smaller instances, $\ModelHeur$ is a bit slower, especially due to the overhead caused by performing the feasibility check (solving \ModelILPFeasibility{}) multiple times. Even though the \ModelILPIdleMax{} scales the best, it fails to produce thermally efficient schedules, as shown in \cref{sec:experiments-comparison}.

\begin{figure}
  \centering
  \newcommand{\Data}{data/times.csv}

\begin{tikzpicture}

    \begin{axis}[width=12cm,
        height=5cm,
        ymin=0.0005,
        grid=major,    
        xtick=data,
        xticklabels from table={\Data}{n},
        table/col sep=comma,
		legend pos=south east,
        legend style={
            font=\scriptsize,
        },		
        legend columns=2,
        xlabel=\footnotesize Number of tasks $n$,
        ylabel=\footnotesize {Time [\si{\second}]},
		ticklabel style = {font=\scriptsize},
		ymode=log,
		ytick={1e-3,1e-2,1e-1,1e0,10,100,1000},
		legend cell align={left},
]


   \addplot[color=Clr1,thick,mark=x] table [x expr=\coordindex,y={HEUR_t-avg}, col sep=comma] {\Data}; \addlegendentry{HEUR};
  \addplot[color=Clr2,thick,mark=x] table [x expr=\coordindex,y={ILP-IDLE-MAX_t-avg}, col sep=comma] {\Data}; \addlegendentry{ILP-IDLE-MAX}

	\addplot[color=Clr3,thick,mark=x] table [x expr=\coordindex,y={ILP-IDLE-MIN_t-avg}, col sep=comma] {\Data}; \addlegendentry{ILP-IDLE-MIN}

   \addplot[color=Clr4,thick,mark=x] table [x expr=\coordindex,y={ILP-SM_t-avg}, col sep=comma] {\Data};   \addlegendentry{ILP-SM}

   \addplot[color=Clr5,thick,mark=x] table [x expr=\coordindex,y={QP-LR-UB_t-avg}, col sep=comma] {\Data}; 
\addlegendentry{QP-LR-UB}

	\addplot [Clr6, thick, mark=o ]  coordinates {(0,300) (1,300) (2,300) (3,300) (4,300) (5,300) (6,300) (7,300) (8,300) (9,300) (10,300) (11,300)};
\addlegendentry{BB-[SM/LR]}

\draw[dashed,gray,line width=2,on layer=axis tics] (-10,5.703)--(120,5.703);
\addplot [dashed,gray,line width=2]  coordinates {(0,300)}; \addlegendentry{Solver time limit}


    \end{axis}

\end{tikzpicture}
    \caption{Average computation time of different methods w.r.t. the instance size $\TaskNum$.}
    \label{fig:scalability}
\end{figure}
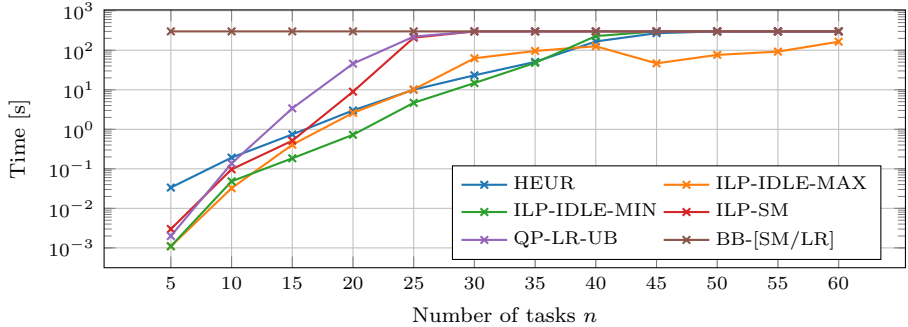

\subsection{Evaluation Summary}

To summarize the results, the linear regression-based power model (\texttt{LR}) exhibited lower errors than the empirical Sum-Max model (\texttt{SM}), but it proved to be harder to integrate with the optimization methods. Its simplified variant \texttt{LR-UB} failed to provide a tight upper bound and, therefore, performed rather poorly.

Considering the optimization methods, global \ModelILPSMOrig{} based on the integer linear programming and simpler \ModelSumMax{} power model provided overall best results. The black-box approach based on metaheuristics proved to be competitive as well; especially, it might be preferred for large-size instances, for which the integer linear programming fails to deliver high-quality solutions in a reasonable time. Also, the \texttt{BB} approach is based on an open-source implementation of a genetic algorithm, which might be an advantage when compared to the other tested methods based on the commercial Gurobi solver.

\section{Conclusion} \label{sec:conclusion}

To conclude, this work studied the problem of offline task allocation on a heterogeneous multi-core platform for safety-critical avionics applications with the aim of minimizing the platform’s steady-state temperature. 
The main limitations are the unavailability of DVFS due to safety requirements and the necessity to schedule the tasks into windows.
Three power models were compared, including the empirical sum-max model, the linear regression model, and its simplified variant providing the upper bound. 
Furthermore, their integration within the optimization procedures was discussed.
Several optimization approaches, including those based on mathematical programming, and both informed and uninformed heuristics, were described and evaluated on three hardware platforms. The data and source code are publicly available\footnote{\url{https://github.com/benedond/safety-critical-scheduling}}.

Extensive experimental evaluation showed that the best-performing method \ModelILPSMOrig{} reduces the platform temperature by up to \SI{16}{\percent}, \SI{14}{\percent} and \SI{10}{\percent} for I.MX8 Ixora, I.MX8 MEK and TX2, respectively, when compared to the worst method \ModelILPIdleMax{}. Moreover,  \ModelILPSMOrig{} saves up to \SI{4.7}{\percent}, \SI{4.6}{\percent}, and \SI{1.8}{\percent} when comparing to the heuristic method \ModelHeur{}, which is a simple model-free approach.
Furthermore, the second best method \ModelBlackBoxSM{}, which does not need an expensive ILP solver, provides better scalability and just \SI{3}{\percent} higher temperatures on average.

Surprisingly, all the best methods rely on the \ModelSumMax{} power model, which has the mean absolute error by \SI{66}{\percent} higher (in average) than the \ModelLR{} model. This shows the importance of designing the optimization method in harmony with the power model.

In our future work, we would like to develop new, scalable algorithms that can handle a larger number of tasks. We were not successful with a branch-and-price method due to its extensive branching. Nevertheless, we aim to develop a logic-based Bender's decomposition~\cite{2019:Hooker} where a master problem would suggest the lengths of intervals (their order is not important, so their ordering may be added as a symmetry-breaking constraint). A sub-problem would be solved as the fixed version of the \textit{ARINC problem} solvable in polynomial time by the Minimum Cost Flow as suggested in~\cref{sssec:complexity}.
This approach needs further investigation, development, and evaluation since the real power models and criterion functions are more complex than the one defined in the \textit{ARINC problem}, but it is definitely worth trying.  

\section*{Acknowledgments}
\label{sec:acknowledgments}

We thank the anonymous reviewers for their inspirational questions and remarks. We thank Claire Hanen from Sorbonne Université, LIP6 laboratory, for her substantial help on the complexity proofs. This work was co-funded by the European Union under the project ROBOPROX (reg. no. CZ.02.01.01/00/22\_008/0004590) and the Grant Agency of the Czech Republic under the Projects GACR 25-17904S and GACR 22-31670S.

\clearpage

\appendix
\section{Performance Characteristics of Benchmarking Kernels} \label{app:ips}

\begin{scriptsize}
\sisetup{round-mode=places, round-precision=2}

\begin{center}
\begin{longtable}{rrrrrrr}
    \caption{Iterations Per Second (IPS) of used kernels.} \label{tab:kernel-ips} \\
        \toprule
            & \multicolumn{2}{c}{I.MX8 MEK} & \multicolumn{2}{c}{I.MX8 Ixora} & \multicolumn{2}{c}{TX2 Developer Kit} \\
            \cmidrule(r){2-3} \cmidrule(lr){4-5} \cmidrule(l){6-7}
         & \multicolumn{1}{c}{little (A53)} & \multicolumn{1}{c}{big (A72)} & \multicolumn{1}{c}{little (A53)} & \multicolumn{1}{c}{big (A72)} & \multicolumn{1}{c}{little (A57)} & \multicolumn{1}{c}{big (Denver)} \\
         & \scriptsize \SI{1200}{\mega\hertz} & \scriptsize \SI{1600}{\mega\hertz}
         & \scriptsize \SI{1200}{\mega\hertz} & \scriptsize \SI{1600}{\mega\hertz}
         & \scriptsize \SI{2035}{\mega\hertz} & \scriptsize \SI{2035}{\mega\hertz} \\
         \midrule
         \begin{filecontents*}{data/ips.csv}
benchmark,ipsimxlow,ipsimxhigh,ipsimxblow,ipsimxbhigh,ipstxlow,ipstxhigh
a2time-4K,32612.423673768575,56012.335798688946,32392.466720951274,55925.01746608948,64788.86511631867,90732.48776772726
a2time-4M,32.10159716781787,54.597102673780576,31.594295155256503,54.51240330745413,58.045907928121274,117.38351858292798
aifirf-4K,3179.217556355342,8806.251084180622,3152.1537573053383,8800.540743038584,10512.202382350515,15104.388032635718
aifirf-4M,3.1118337852463274,8.750770614737261,3.0941446826425767,8.743126976159434,10.386635319289176,15.239056340623314
bitmnp-4K,11763.412319745046,24485.253853775186,11655.709788528418,24536.22268599553,27287.51655727213,47289.57916251991
bitmnp-4M,9.844230698134272,21.652390033609212,9.761152803724867,21.541788408765548,20.199688661458005,34.08150767512244
canrdr-4K,33409.004086958164,62938.09802662973,33078.73548937734,63262.987381231054,163382.50519027896,189621.55071905692
canrdr-4M,35.574430617561944,68.12068169140876,35.20040505810108,68.01731975607113,182.91159077406087,226.16485533246603
idctrn-4K,6966.355634806228,20301.26581558189,6914.018296227593,20269.733339383984,24190.641138403116,28626.422088756855
idctrn-4M,6.860659050756539,20.493198176503125,6.812828870152626,20.469020766241012,24.235553618015796,28.857141702953058
iirflt-4K,8253.335748777794,23629.724781021327,8178.397860731265,23748.03276144869,23785.58757314263,35237.472995596385
iirflt-4M,8.228420317158772,18.833389498312336,8.166007084664427,18.806409510487264,22.54348665409052,35.31187069895487
matrix-4K,4029.8592714738775,10552.82603937601,4002.3179123047194,11235.613721534904,12562.776144866519,9796.331527268949
matrix-4M,14.76056026336543,41.731576796970025,14.599406197333172,41.654266824085305,43.02468988299324,34.849093423117914
membench-1K-RO-R,11.88500198866914,11.873568230412276,11.795801531969717,11.869620881640216,14.92882393582086,20.03564957900426
membench-1K-RO-S,11.884901779795072,11.872806368502618,11.79686173843774,11.86867909027156,14.932565049771423,20.041201990677173
membench-1K-RW-R,4.4576294408106545,7.915482297611162,4.4241185782645855,7.912174212180628,14.936065346947244,19.998044786792782
membench-1K-RW-S,4.457650159768861,7.91611441728846,4.423850040631812,7.913412541375117,14.934044353390046,19.991407648982065
membench-1M-RO-R,2.113017276701274,1.629822656955235,1.69858412113141,1.5027195257441273,2.085915712585311,2.441307160280762
membench-1M-RO-S,4.228943181524877,3.3202418712696127,3.9884222182508733,3.1662906807815543,6.17705590683107,5.021953746090861
membench-1M-RW-R,1.4182720686963894,1.2746391511617854,1.152031786218051,1.2885382720098002,1.8854276181930447,2.5625970715195288
membench-1M-RW-S,3.7802658242186205,2.9385315957261735,3.4945149045036272,2.801756786325792,5.273358400275865,4.6777587113696955
membench-4M-RO-R,0.15207581587724744,0.19523999413868984,0.14071484514159688,0.17942960983307388,0.2301531726902547,0.1812903749485161
membench-4M-RO-S,0.7058449372280333,2.4248728304502465,0.6299814222858118,2.0144994045650906,1.492349913198103,4.246879931332492
membench-4M-RW-R,0.1450280249392954,0.18951909117119112,0.13456851124440677,0.17183153306581683,0.22900194050092088,0.17947152882649264
membench-4M-RW-S,0.6252644137314116,1.5008836995974255,0.4810618562454629,1.3553806170813016,1.4640198520289731,3.6856330847504015
pntrch-4K,279.3928104436285,773.8777047067479,277.12661185685477,815.5385599172752,843.5510111093831,1053.4272506224622
pntrch-4M,0.2700828336961272,0.791742863565567,0.2679610794570921,0.7920639205610019,1.0148986009583056,1.4381310815112005
puwmod-4K,18709.9642615774,35919.98421069177,18547.45884223171,35863.73348178057,45702.325222492276,79401.36533174082
puwmod-4M,19.03183842594222,37.414736950664256,18.856854576243933,37.382727130691116,47.51072076988449,87.53314316762435
rspeed-4K,53498.403844058616,136629.64552691436,52959.45904804554,136940.3600303121,171090.58005379437,140909.25768147424
rspeed-4M,56.069853394657954,168.75376308360455,54.88332701550384,168.47017273530935,198.18418048629317,188.11147297777194
tblook-4K,19232.011529008505,59599.439577842466,19025.705169697823,60082.12458349695,73168.13522355782,82280.65755095085
tblook-4M,19.589798006080102,63.66390732325725,19.333508827323993,62.5012573364569,73.65013990111893,86.65432075235475
tinyrenderer-boggie,3.2508907221833545,6.327494421130616,3.1674550597095608,6.231748891075363,5.669385253405165,6.255116259906406
tinyrenderer-diablo,2.850456719948205,5.719011102653788,2.810369156776637,5.68037325148235,5.01212868312935,5.6819673406730375
ttsprk-4K,67723.00597293641,138148.80000317306,67350.74009369145,137949.1575933522,172158.44300169378,242422.25519398708
ttsprk-4M,183.358916522343,390.4910125105879,180.59668785140366,388.7248201869689,483.1691639665602,645.5043330255228
\end{filecontents*}
\csvreader[head to column names,late after line=\\]{data/ips.csv}{}{\benchmark & \num{\ipsimxlow} & \num{\ipsimxhigh} & \num{\ipsimxblow} & \num{\ipsimxbhigh} & \num{\ipstxlow} & \num{\ipstxhigh}}
         \bottomrule
\end{longtable}
\end{center}
\end{scriptsize}

\clearpage

\section{Task Characteristic Coefficients} \label{app:kernel-chars-values}

\begin{scriptsize}
\sisetup{round-mode=places, round-precision=2}

\begin{center}
\begin{longtable}{rrrrrrr}    
    \caption{Task characteristics coefficients identified for used kernels on little cluster.} \label{tab:kernel-char-little} \\
    \toprule
    & \multicolumn{2}{c}{I.MX8 MEK} & \multicolumn{2}{c}{I.MX8 Ixora} & \multicolumn{2}{c}{TX2 Developer Kit} \\
    \cmidrule(r){2-3} \cmidrule(lr){4-5} \cmidrule(l){6-7}
     & \scriptsize $\TaskCoefOffset{\TaskIdx}{1}$ [\si{\watt}] & \scriptsize $\TaskCoefSlope{\TaskIdx}{1}$ [\si{\watt}]
     & \scriptsize $\TaskCoefOffset{\TaskIdx}{1}$ [\si{\watt}] & \scriptsize $\TaskCoefSlope{\TaskIdx}{1}$ [\si{\watt}]
     & \scriptsize $\TaskCoefOffset{\TaskIdx}{1}$ [\si{\watt}] & \scriptsize $\TaskCoefSlope{\TaskIdx}{1}$ [\si{\watt}]
      \\
     \midrule
     \begin{filecontents*}{data/parameters_little.csv}
benchmark,affinity,slopeimxa,interceptimxa,slopeimxb,interceptimxb,slopetx,intercepttx
a2time-4K,little,0.2497232537391287,0.2476346997481329,0.2418995207650015,0.2444344560578661,0.5800045045045045,0.2926846846846844
a2time-4M,little,0.3524115692386657,1.3782962421469795,0.3152722258653364,0.9519121379116776,0.6962015765765767,0.5099808558558561
aifirf-4K,little,0.2632852733690904,0.1996395430486366,0.2383305209708416,0.2340051570004027,0.8408806306306307,0.3448423423423423
aifirf-4M,little,0.3679626314253736,0.9496703835860236,0.3209989693391038,0.6573638540169515,0.8924351045961215,0.5116324629714457
bitmnp-4K,little,0.2202160897805022,0.2110548830954544,0.1966626197979089,0.2453778574288883,0.4633817567567566,0.2764324324324327
bitmnp-4M,little,0.2925454420827374,1.202637172383569,0.266294329175358,0.8281220943565062,0.4893648648648647,0.4969121621621624
canrdr-4K,little,0.3721560685005431,0.1606655111173918,0.3299832408326685,0.2315548473155226,0.7469335585585587,0.3394515765765762
canrdr-4M,little,0.4149180687464311,1.360403835101823,0.3755011537810582,0.9300927164287804,1.0026655405405405,0.5859493243243246
idctrn-4K,little,0.2760194204433095,0.2253277488598239,0.2585078455955605,0.2261819160893807,0.8571373873873872,0.3621666666666669
idctrn-4M,little,0.3536776687146732,1.1517181389720026,0.3142554075379395,0.7711385046843713,0.919427927927928,0.4975416666666663
iirflt-4K,little,0.2467029096159461,0.2061375774456646,0.2295052855909555,0.2214991240550246,0.7663175675675676,0.3278986486486488
iirflt-4M,little,0.3165004134693807,1.1974815675990076,0.2899535633660057,0.7997804581462562,0.8040743243243244,0.5851216216216217
matrix-4K,little,0.3460477873921794,0.2171453969626338,0.3335254213993099,0.2015414247750238,0.920793918918919,0.3325608108108104
matrix-4M,little,0.512836061363552,0.6491392956255408,0.4313767949270035,0.4630255145239212,0.9407590090090088,0.4328693693693699
membench-1K-RO-R,little,0.1680394334163496,0.2200243323134385,0.1473708308040073,0.2286307707438517,0.3445947396549089,0.2462764467857692
membench-1K-RO-S,little,0.1644425390553952,0.2357333411872515,0.1521274221990604,0.2249168894607027,0.3378614864864864,0.233695945945946
membench-1K-RW-R,little,0.1950024176908522,0.2523965863181221,0.1835987645105898,0.2103481325319691,0.6157792792792794,0.2934977477477476
membench-1K-RW-S,little,0.1924542224772191,0.2293743070758669,0.1808139097961924,0.2072226617054768,0.6266285730645902,0.3066585890975717
membench-1M-RO-R,little,0.5720731046089496,0.370614442192446,0.3539638540109318,0.5741854404433324,0.4055427927927928,0.3237274774774774
membench-1M-RO-S,little,0.7694183493303006,0.548891603708654,0.5319358817642469,0.5728396186226359,0.5579065315315316,0.4031914414414412
membench-1M-RW-R,little,0.629845657117221,0.4839350507524615,0.3987001775162436,0.5739344515756954,0.5357117117117117,0.2854099099099096
membench-1M-RW-S,little,0.8994551790858697,0.3966595852622356,0.5292008467237537,0.6019881339398916,0.8781632882882883,0.2815259009009008
membench-4M-RO-R,little,0.2404901207821013,1.8063055698255548,0.1914517590485838,1.2297366956155953,0.3683040540540541,0.4751790540540535
membench-4M-RO-S,little,0.4460193219495461,1.811811973946786,0.3629087591866122,1.2727281382678983,0.4828355855855857,0.6854279279279272
membench-4M-RW-R,little,0.3100059882180295,1.780091640146768,0.22969553581972,1.242050724450058,0.4905585585585585,0.4941948198198198
membench-4M-RW-S,little,0.4809923394965724,2.021926050715673,0.3449902757639942,1.373691633717554,0.6489887387387387,1.2082747747747749
pntrch-4K,little,0.2639037027402671,0.226702564251525,0.2387831493354596,0.2073344981839175,0.761179054054054,0.3995270270270272
pntrch-4M,little,0.4137540761596208,0.4395494898758629,0.3437598877732618,0.314072044297375,0.7678231981981979,0.3952882882882891
puwmod-4K,little,0.3186729670708705,0.2543404557478919,0.2904875493305144,0.2164994822680732,0.691563063063063,0.3207342342342345
puwmod-4M,little,0.3646338012683732,1.3576739744676312,0.3260835424599421,0.90646297292815,0.813394144144144,0.5912038288288288
rspeed-4K,little,0.379011317016731,0.1871136569598919,0.3265305224631427,0.1981941350646083,0.8402195945945944,0.4135168918918923
rspeed-4M,little,0.4878887500707751,1.5201358310627775,0.4266070197109079,1.037703453942644,1.2284369369369366,0.9102927927927936
tblook-4K,little,0.2480810187585183,0.2091083493165904,0.2384487972855504,0.2032701812455615,0.7996722972972975,0.4000540540540536
tblook-4M,little,0.2637937776105152,1.330340057896417,0.266576550498673,0.8215447189578793,0.8991126126126124,0.4312522522522526
tinyrenderer-boggie,little,0.2923976448944616,1.3675635362510563,0.2648500424149172,0.901486860203268,0.5872038288288287,0.4857015765765764
tinyrenderer-diablo,little,0.276058562153253,1.3260525809749248,0.2581147359399987,0.8601760587883769,0.5802162162162162,0.4556554054054054
ttsprk-4K,little,0.3986020676568633,0.1896114346462001,0.3434870032582224,0.1930229223856399,0.6713400900900902,0.3916632882882878
ttsprk-4M,little,0.7086412536721273,-0.1577432751951866,0.484642344640702,0.1617381599298246,0.7108918918918921,0.3822972972972969
\end{filecontents*}
\csvreader[head to column names,late after line=\\]{data/parameters_little.csv}{}{\benchmark & \num{\interceptimxa} & \num{\slopeimxa} & \num{\interceptimxb} & \num{\slopeimxb} & \num{\intercepttx} & \num{\slopetx}}
     \bottomrule
\end{longtable}
\end{center}
\end{scriptsize}

\newpage

\begin{scriptsize}
\sisetup{round-mode=places, round-precision=2}
\begin{center}
\begin{longtable}{rrrrrrr}    \caption{Task characteristics coefficients identified for used kernels on big cluster.} \label{tab:kernel-char-big} \\
    \toprule
    & \multicolumn{2}{c}{I.MX8 MEK} & \multicolumn{2}{c}{I.MX8 Ixora} & \multicolumn{2}{c}{TX2 Developer Kit} \\
    \cmidrule(r){2-3} \cmidrule(lr){4-5} \cmidrule(l){6-7}
     & \scriptsize $\TaskCoefOffset{\TaskIdx}{2}$ [\si{\watt}] & \scriptsize $\TaskCoefSlope{\TaskIdx}{2}$ [\si{\watt}]
     & \scriptsize $\TaskCoefOffset{\TaskIdx}{2}$ [\si{\watt}] & \scriptsize $\TaskCoefSlope{\TaskIdx}{2}$ [\si{\watt}]
     & \scriptsize $\TaskCoefOffset{\TaskIdx}{2}$ [\si{\watt}] & \scriptsize $\TaskCoefSlope{\TaskIdx}{2}$ [\si{\watt}]
      \\
     \midrule
     \begin{filecontents*}{data/parameters_big.csv}
benchmark,affinity,slopeimxa,interceptimxa,slopeimxb,interceptimxb,slopetx,intercepttx
a2time-4K,big,1.0505038527520547,0.1590802978220624,0.9844350429996958,0.1908365467570911,1.510864864864864,0.3625439189189205
a2time-4M,big,1.142446544783514,1.627796292441822,1.079751949123648,1.096376591296715,2.0528716216216214,0.4960878378378379
aifirf-4K,big,1.467876146533304,0.1396438598989302,1.386087803783588,0.1746077060316659,1.8105270270270264,0.3804729729729739
aifirf-4M,big,1.7229855953111368,1.0822298032629174,1.6126508678631255,0.703821472286454,1.890743243243244,0.482195945945945
bitmnp-4K,big,0.7441056086603632,0.2228277105969223,0.7207295540108083,0.2147223384945142,1.4937770270270274,0.3463581081081073
bitmnp-4M,big,0.8810579510664649,1.3116924335644953,0.7927847647671751,0.927864101136218,1.4536959459459462,0.5314358108108101
canrdr-4K,big,1.0818332458759352,0.2948354185489847,1.1068858215957862,0.1447493941499278,1.841902027027026,0.3686351351351367
canrdr-4M,big,1.3203872798977994,1.5866328796525693,1.264783732754613,1.0470248708009935,2.232209459459459,0.5472162162162166
idctrn-4K,big,1.5091972932995956,0.1866744260956192,1.5013160507991206,0.1040936230163511,2.109983108108107,0.3624898648648664
idctrn-4M,big,1.6805007561794163,1.2741246873310423,1.6541146586528876,0.7102082756766999,2.2088547297297296,0.4898918918918924
iirflt-4K,big,1.3097385946926767,0.2987186713822964,1.2803810129186657,0.2073526660699398,1.3804966216216217,0.417351351351352
iirflt-4M,big,1.4667992273387576,1.4311409769043966,1.437909073879684,0.8532361965874804,1.5940912162162162,0.5391081081081084
matrix-4K,big,1.3037322942476282,0.9018167543654076,1.59400421889507,0.4107467942697651,1.7850506756756754,0.3750202702702707
matrix-4M,big,1.9204293396101288,0.987126913701328,1.78741671943967,0.5854096624956009,1.8287533783783785,0.3829695945945954
membench-1K-RO-R,big,0.4129014750198401,0.2490598554188983,0.4063945077663469,0.2050916645533993,0.5994222972972976,0.3723445945945944
membench-1K-RO-S,big,0.4189357264462323,0.2490857555281529,0.4089263989358809,0.1922280862366046,0.5848141891891894,0.380388513513513
membench-1K-RW-R,big,0.8270718841032094,0.2233031404898815,0.7785204640352088,0.1786936098548599,1.3402452817224004,0.4171513284470909
membench-1K-RW-S,big,0.8068135070992133,0.2083540854217158,0.7678346614769094,0.174769089466289,1.3578243243243246,0.3924121621621617
membench-1M-RO-R,big,0.6714747303011528,1.4835507913611783,0.5165982881210658,1.046828349198479,0.2752229729729727,0.9729932432432437
membench-1M-RO-S,big,1.244531363451368,1.2202345759831648,0.8970942789843752,0.9613932013879536,1.87335472972973,-0.2659932432432437
membench-1M-RW-R,big,0.8074583342028259,1.5781717454317103,0.5535849850119448,1.154812684247687,0.599618243243242,1.0122297297297318
membench-1M-RW-S,big,1.0107583060971503,1.8622928474573768,0.6578410535418335,1.4281025217045942,2.855277027027027,-0.7494290540540534
membench-4M-RO-R,big,0.5126996133532016,1.8001058247985338,0.3686624440002752,1.3298289281791955,0.5425878378378379,0.3908614864864863
membench-4M-RO-S,big,0.2629796122508807,3.1755975883402723,0.2956825000977829,2.1910576524402297,1.6576993243243248,0.8996418918918914
membench-4M-RW-R,big,0.5541057890003449,1.900764028157484,0.4276393244857503,1.3746744877888712,0.6722432432432432,0.4272398648648647
membench-4M-RW-S,big,0.3468575095340505,3.1745930327581267,0.2909001557809798,2.347726106844253,2.0213817567567567,1.6273074324324326
pntrch-4K,big,1.2498996985359554,0.3853134094351986,1.3315865970583562,0.1042311916423939,1.583672297297297,0.3158412162162167
pntrch-4M,big,1.5876272330254544,0.740877751927389,1.4711299970674556,0.4043696829001826,1.5450709459459464,0.4365236486486483
puwmod-4K,big,1.1902371682446802,0.2378079962607042,1.1360236409863171,0.1747826830162901,1.803689189189189,0.3690304054054061
puwmod-4M,big,1.3890882841381504,1.5298611285612296,1.3528106385325094,0.9485877787580276,1.9898378378378376,0.524418918918919
rspeed-4K,big,1.6483098890417338,0.1469585441875551,1.4654431242527863,0.1429253323811581,1.235601408612002,0.3943816422354556
rspeed-4M,big,2.207915565695288,1.7580574625694894,2.0753341421914677,1.1368483351631644,1.7044628378378377,0.7454932432432435
tblook-4K,big,1.4983469480133706,0.2878774520654144,1.4522183643887034,0.1516929252501606,1.4760168918918923,0.3817195945945939
tblook-4M,big,1.678870028611744,1.3523695804929048,1.616966429225399,0.7814475802311316,1.606858108108108,0.4698817567567572
tinyrenderer-boggie,big,1.0853963898218346,1.46820598092752,0.9626852080259498,1.0225416202228148,0.9795777027027032,0.5237668918918916
tinyrenderer-diablo,big,1.1353433539191018,1.3381732450780175,1.0001108637834086,0.9289230518376934,0.9079121621621624,0.5212364864864862
ttsprk-4K,big,1.2940971186259231,0.2048468448157194,1.1495060746323045,0.2359049299257094,1.4649256756756754,0.4056689189189195
ttsprk-4M,big,2.05951388526124,-0.3373093827503854,1.577697531704464,0.0545384037012022,1.5038783783783778,0.3456824324324334
\end{filecontents*}
\csvreader[head to column names,late after line=\\]{data/parameters_big.csv}{}{\benchmark & \num{\interceptimxa} & \num{\slopeimxa} & \num{\interceptimxb} & \num{\slopeimxb} & \num{\intercepttx} & \num{\slopetx}}
     \bottomrule
\end{longtable}
\end{center}
\end{scriptsize}

\clearpage

\section{Transformation of Continuous Variables to a Feasible Allocation} \label{app:reconstruction-algorithm}

\begin{algorithm}
\algsetup{linenosize=\scriptsize}
\footnotesize
  
\SetKwInOut{Input}{input}
\SetKwInOut{Output}{output}
\SetKwProg{Fun}{Function}{ is}{end}
\Input{instantiation of variables $x_{\TaskIdx} \ \forall \Task{\TaskIdx} \in \TaskSet$}
\Output{Functions mapping tasks to windows and tasks to clusters (or failure)}

$TaskAssignmentPreference(\TaskIdx) \gets \frac{x \% \frac{1}{\ResNum}}{\frac{1}{\WinNum \cdot \ResNum}} \quad \forall \Task{\TaskIdx} \in \TaskSet$

$TaskAssignmentCluster(\TaskIdx) \gets \left\lfloor \frac{x}{\ResNum} \right\rfloor + 1 \quad \forall \Task{\TaskIdx} \in \TaskSet$

$TaskAssignmentWindow(\TaskIdx) \gets \left\lfloor TaskAssignmentPreference(\TaskIdx) \right\rfloor + 1 \quad \forall \Task{\TaskIdx} \in \TaskSet$

$WindowCapacity(\WinIdx,\ResIdx) \gets \ResCap{\ResIdx} \quad \forall \Win{\WinIdx} \in \WinSet, \Res{\ResIdx} \in \ResSet$

$TasksInWindow(\WinIdx) \gets \{\} \quad \forall \Win{\WinIdx} \in \WinSet$

$WindowLength(\WinIdx) \gets 0 \quad \forall \Win{\WinIdx} \in \WinSet$

$TaskAssigned(\TaskIdx) \gets False \quad \forall \Task{\TaskIdx} \in \TaskSet$

\For{$ Iteration \in \{0,2,\dots,2 \cdot \WinNum - 1 \} $}{
    $CurrentWindow \gets (Iteration\, \% \, \WinNum) + 1 $

    $TasksToCurrentWindow \gets \{\TaskIdx \ | \ \Task{\TaskIdx} \in \TaskSet \wedge  TaskAssignmentWindow(\TaskIdx) = CurrentWindow \wedge \neg TaskAssigned(\TaskIdx) \}$
    
    sort $TasksToCurrentWindow$ by $TaskAssignmentPreference$ in non-decreasing order
    
    \For{$ \TaskIdx \in TasksToCurrentWindow$}{
        \If{$WindowCapacity(\WinIdx, TaskAssignmentCluster(\TaskIdx)) > 0$}{
            $WindowCapacity(\WinIdx,TaskAssignmentCluster(\TaskIdx)) \mathrel{{-}{=}} 1$
            
            $TasksInWindow(\WinIdx) \gets TasksInWindow(\WinIdx) \cup \{\TaskIdx\}$
            
            $TaskAssigned(\TaskIdx) \gets True$
            
            $WindowLength(\WinIdx) \gets \max\{ WindowLength(\WinIdx), \TaskProc{\TaskIdx}{TaskAssignmentCluster(\TaskIdx)} \}$
        }
        \Else{
            \If{$CurrentWindow = \WinNum - 1$}{
                $TaskAssignmentWindow(\TaskIdx) \gets CurrentWindow + 1$
            }
            \Else{
                $TaskAssignmentWindow(\TaskIdx) \gets (CurrentWindow + 1) \, \% \, \WinNum$
            }
            $TaskAssignmentPreference(\TaskIdx) \gets 0$
        }
    }
    
}

\If{$\sum\limits_{\Win{\WinIdx} \in \WinSet} WindowLength(\WinIdx) > \MF \vee \neg \bigwedge\limits_{\Task{\TaskIdx} \in \TaskSet} TaskAssigned(\TaskIdx)$}{\Return{Schedule reconstruction failed.}}
\Else{
    \Return{($TaskAssignmentWindow, TaskAssignmentCluster$)}
}

 \caption{Get a feasible allocation from the instantiation of variables $x_{\TaskIdx}$, or report a failure.}
 \label{alg:reconstruction} 
\end{algorithm}

\clearpage

\section{Example Schedules} \label{app:example-schedules}

\begin{figure}[htb]
    \centering
    \input{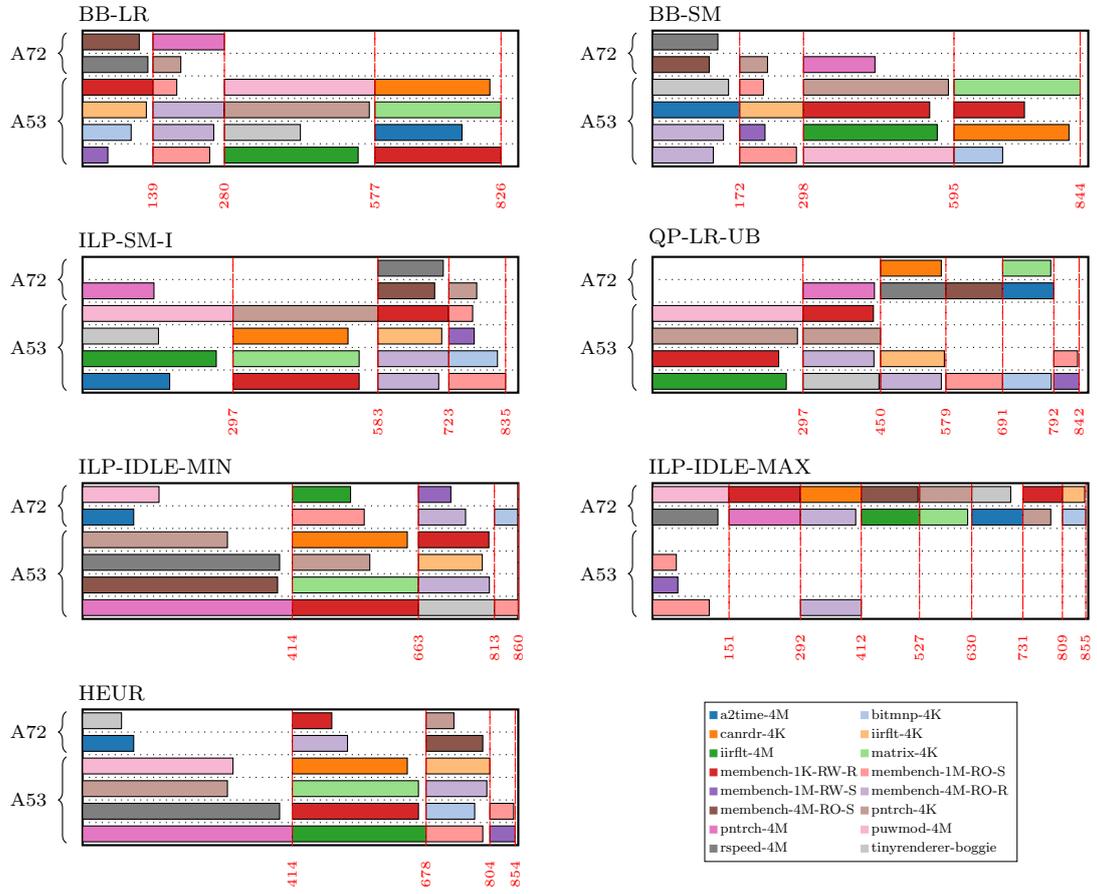}
    \caption{Example schedules for instance no. 1 on I.MX8~MEK (mixed-workload).}
    \label{fig:examples-schedules}
\end{figure}

\clearpage

\bibliography{mybibfile}

\end{document}